\documentclass{article}
\pdfoutput=1
\usepackage{euscript,amsmath,amssymb,amsfonts,graphicx,bm,amsthm}
\usepackage[top=1in, bottom=1in, left=1in, right=1in]{geometry}

  \usepackage{paralist}
  \usepackage{graphics}
\usepackage[caption=false]{subfig}
\usepackage{soul}

\usepackage{latexsym,amssymb,enumerate,amsmath,amsthm}
  \usepackage{graphicx}
  \usepackage{tikz}
     \usepackage[colorlinks=false]{hyperref}
 \hypersetup{urlcolor=blue, citecolor=red}
\usepackage{latexsym,amssymb,enumerate,amsmath}
\usepackage{pgfplots}
\usepackage{soul,xcolor}

\usepackage{mathtools}

\newtheorem{theorem}{Theorem}

\newtheorem{lemma}[theorem]{Lemma}
\theoremstyle{plain}
\theoremstyle{remark}

\theoremstyle{definition}
\newtheorem*{definition*}{Definition}

\begin{document}

\title{Exact distributions for stochastic models of gene expression with arbitrary regulation}


\author{Zihao Wang, Zhenquan Zhang, Tianshou Zhou\thanks{School of Mathematics, Sun Yat-sen University, Guangzhou {\rm510275}, China. E-mail: {mcszhtsh@mail.sysu.edu.cn}}
}

\date{\today}
\maketitle

\begin{abstract}
Stochasticity in gene expression can result in fluctuations in gene product levels. Recent experiments indicated that feedback regulation plays an important role in controlling the noise in gene expression. A quantitative understanding of the feedback effect on gene expression requires analysis of the corresponding stochastic model. However, for stochastic models of gene expression with general regulation functions, exact analytical results for gene product distributions have not been given so far. Here, we propose a technique to solve a generalized ON-OFF model of stochastic gene expression with arbitrary (positive or negative, linear or nonlinear) feedbacks including posttranscriptional or posttranslational regulation. The obtained results, which generalize results obtained previously, provide new insights into the role of feedback in regulating gene expression. The proposed analytical framework can easily be extended to analysis of more complex models of stochastic gene expression.
\end{abstract}

\section{Introduction}

Gene expression is a complex process: Apart from fundamental sub-processes such as transcription and translation described by the central dogma in biology, it also involves other sub-processes such as switching between promoter activity states, stochastic partitioning at cell division \cite{s13}, feedback regulation, and posttranscriptional or posttranslational regulation. Since these sub-processes are biochemical, fluctuations (or the noise) in the levels of gene products (mRNA and protein) are inevitable, implying that gene expression is inherently noisy. This molecular noise (also called cell-to-cell variability in gene expression) can carry out important biological functions. For example, in unicellular organisms, the noise can improve fitness by inducing phenotypic differences within a population of genetically identical cells, enabling a rapid response to a fluctuating environment and thus enhancing the chance of cell survival in this environment \cite{s2,s4,s3,s24,s30,s38}. Also for example, in multi-cellular organisms, the noise plays an important role in development, e.g., it allows identical progenitor cells to acquire distinct phenotypes for better survival \cite{s8,s31}. Because of the functional importance of molecular noise, an important task in the post-genome era is to understand how different regulatory mechanisms control variations in mRNA and protein levels across a clone population of cells. Quantifying the impact of gene expression noise using stochastic models is also an important step towards understanding intracellular processes \cite{s9,s10,s17,s19,s26,s27,s28,s18}.

Although a variety of factors can affect gene expression levels in different ways, experimental measurements support two kinetic modes of gene expression: the constitutive mode in which gene products are synthesized in stochastic and uncorrelated events \cite{s22,s40}, and the bursty mode in which gene products are generated in a manner of high activity followed by a long refractory period \cite{s12,s25,s29}. Moreover, the latter mode is more common than the former mode in prokaryotic cells \cite{s5,s32,s39}. Single cell measurements have provided evidence for transcriptional or translational bursting (i.e., production of mRNAs or proteins in bursts) \cite{s7,s12,s29}. Although the molecular sources of generating bursts remain poorly understood \cite{s6}, several lines of evidence \cite{s6,s21,s34,s37} have pointed to switching between active (ON) and inactive (OFF) promoter states as an important source of gene expression noise, which is responsible for generating heterogeneity in the response of isogenic cells to the same stimulus. It has been demonstrated in yeast cells that high levels of cell-to-cell variability, originated by slow promoter state fluctuations, may confer cell colonies with an enhanced probability of survival when subjected to external stresses such as addition of high concentrations of antibiotic \cite{s1}. In this paper, we will adopt the extensively used ON-OFF model of stochastic gene expression for analysis.

As a ubiquitous mechanism of controlling signals, feedback has been identified in various gene regulatory systems in prokaryotic or eukaryotic cells. For example, 40\% of \emph{E. coli} transcription factors negatively self-regulate transcription of their own genes \cite{s33}. It was shown that for simple (e.g., linear) feedback regulation, Paulsson showed that positive feedback amplifies the gene expression noise whereas negative feedback reduces the noise \cite{s26}; subsequently, Hornung and Barkai showed that negative feedback in fact amplifies rather than reduces the noise when parameters are chosen to preserve system sensitivity and if the intrinsic noise is negligible, while positive feedback reduces the noise when susceptibility (i.e., steady state sensitivity) is controlled \cite{s16}. We ever showed that when system sensitivity is maintained, either there exists a minimum of the output noise intensity with a biologically feasible feedback strength, or the output noise intensity is a monotonic function of feedback strength bounded by both biological and dynamical constraints \cite{s41}. In spite of these, we note that the noise used in these works, which is defined as its variance normalized by the square of its mean (noise intensity) or the ratio of the variance over the mean (Fano factor), would not correctly characterize stochastic fluctuations since the underlying distributions may be bimodal \cite{s35} or tail-weighted \cite{s37}.

The statistics and dynamics of stochastic gene expression are best characterized by the probability mass function, $P\left( n;t \right)$, i.e., the probability that there are exactly $n$ mRNA or protein molecules of a gene of interest at time $t$ in a single cell. Previous studies have derived analytical gene product distributions in common two-state model of stochastic gene expression \cite{s20,s29,s32,s36,s42,s43}, or in similar gene models with linear feedback \cite{s14,s15,s23}. However, transcription factors regulate gene expression often in a nonlinear fashion. Moreover, the corresponding regulation functions usually take Hill-type forms \cite{s1}. For two-state models of stochastic gene expression with nonlinear feedback regulation, exact analytical results for gene product distributions have not been obtained so far. This motivates the study of this paper.

Here, we develop a new technique to derive the exact steady-state protein distribution in a generalized ON-OFF model of stochastic gene expression with arbitrary feedbacks, where by ¡®arbitrary¡¯ we mean that feedback regulation may be positive or negative, linear or nonlinear, and even posttranscriptional or posttranslational. The derived distributions provide new insights into the role of feedback in regulating the gene expression noise.

The rest of the paper is organized as follows: Section 2 describes a gene model to be studied and gives its mathematical equation. Section 3 derives the explicit expressions of stationary protein distributions. Section 4 reproduces known protein distributions. And section 5 concludes this paper and gives a brief discussion.

\section{A general gene model and its mathematical equation}
In order to model the bursty expression of a gene, we assume that the gene promoter has one active (ON) state where the gene is expressed and one inactive (OFF) state where the gene is not expressed, and there are stochastic transitions from OFF to ON states and vice versa. Also assume that each mRNA degrades instantaneously after producing a protein molecule, and the produced protein molecules can, as transcription factors, self-regulate the switching rates from ON (OFF) to OFF (ON) states as well as the synthesis rate of the protein. Finally, the produced protein is assumed to degrade in a linear manner with a constant rate.

Denote by $X$ the protein, which is a random variable. Let $n$ represent the number of protein molecules and $\delta $ be the protein degradation rate. Then, under the above assumed conditions, the biochemical reactions for the gene model are listed below
\begin{equation}\label{eq:1}
\begin{split}
 &\text{OFF}\xrightarrow{{{K}_{\text{1}}}\left( n \right)}\text{ON,} \\
  &\text{ON}\xrightarrow{{{K}_{\text{2}}}\left( n \right)}\text{OFF,} \\
  &\text{ON}\xrightarrow{{{K}_{\text{3}}}\left( n \right)}\text{ON}+\text{X,} \\
  &\text{X}\xrightarrow{{{K}_{4}}\left( n \right)}\varnothing , \\
 \end{split}
\end{equation}
where functions ${{K}_{i}}\left( n \right)$ ($1\le i\le 4$), which characterize auto-regulations, should be understood as reaction propensity functions, and in particular, $ {{K}_{4}}\left( n \right)=n\delta $. Without loss of generality, we assume that regulating functions ${{K}_{i}}\left( n \right)$ take Hill-type forms that will be specified. Note that if ${{K}_{1}}\left( n \right)$ is not a constant, this corresponds to positive feedback; if ${{K}_{2}}\left( n \right)$ is not a constant, this corresponds to negative feedback; and if ${{K}_{3}}\left( n \right)$ is not a constant, this corresponds to posttranscriptional or posttranslational regulation. In addition, if all ${{K}_{i}}\left( n \right)$ ($1\le i\le 3$) are constants, the corresponding gene model is just the common ON-OFF model of stochastic gene expression. Therefore, the model described by \eqref{eq:1} includes almost gene models studied in the literature, and is therefore general.

Now, we establish a mathematical model in the sense of the chemical master equation \cite{s18} for the gene expression system described by \eqref{eq:1}. Let ${{P}_{0}}\left( n;t \right)$ and ${{P}_{1}}\left( n;t \right)$ represent the probabilities that the protein has $n$ molecules in OFF and ON states at time $t$, respectively. Assume that all the reaction events involved are Markovian, that is, the probabilities that the reaction events to happen depend only on the present state of the system, independent of the prior history. This hypothesis is made in almost all previous studies. In particular, the famous Gillespie stochastic simulation algorithm \cite{s11} is also based on the hypothesis. Then, the chemical master equation corresponding to reaction \eqref{eq:1} takes the form \cite{s18}
\begin{equation}\label{eq:2}
\begin{aligned}
   \frac{\partial {{P}_{0}}\left( n;t \right)}{\partial t}&=-{{K}_{1}}\left( n \right){{P}_{0}}\left( n;t \right)+{{K}_{2}}\left( n \right){{P}_{1}}\left( n;t \right)+\delta \left( \mathbb{E}-\text{I} \right)\left[ n{{P}_{0}}\left( n;t \right) \right], \\
  \frac{\partial {{P}_{1}}\left( n;t \right)}{\partial t}&={{K}_{1}}\left( n \right){{P}_{0}}\left( n;t \right)-{{K}_{2}}\left( n \right){{P}_{1}}\left( n;t \right)+\left( {{\mathbb{E}}^{-1}}-\text{I} \right)\left[ {{K}_{3}}\left( n \right){{P}_{1}}\left( n;t \right) \right]\\
  &+\delta \left( \mathbb{E}-\text{I} \right)\left[ n{{P}_{1}}\left( n;t \right) \right], \\
   \end{aligned}
\end{equation}
where $\mathbb{E}$ is the common step operator and ${{\mathbb{E}}^{-1}}$ is its inverse, and $\text{I}$ is the unit operator. Assume that the stationary distributions always exist (this has been numerically verified by analyzing a simple example). The steady-state equation corresponding to \eqref{eq:2} reads
\begin{equation}\label{eq:3}
\begin{split}
  & -{{{\tilde{K}}}_{1}}\left( n \right){{P}_{0}}\left( n \right)+{{{\tilde{K}}}_{2}}\left( n \right){{P}_{1}}\left( n \right)+\left( \mathbb{E}-\text{I} \right)\left[ n{{P}_{0}}\left( n \right) \right]=0, \\
 & {{{\tilde{K}}}_{1}}\left( n \right){{P}_{0}}\left( n \right)-{{{\tilde{K}}}_{2}}\left( n \right){{P}_{1}}\left( n \right)+\left( {{\mathbb{E}}^{-1}}-\text{I} \right)\left[ {{{\tilde{K}}}_{3}}\left( n \right){{P}_{1}}\left( n \right) \right]+\left( \mathbb{E}-\text{I} \right)\left[ n{{P}_{1}}\left( n \right) \right]=0, \\
   \end{split}
\end{equation}
where reaction propensity function ${{K}_{i}}\left( n \right)$ is normalized by the degradation rate, that is, ${{\tilde{K}}_{i}}\left( n \right)={{{K}_{i}}\left( n \right)}/{\delta }\;$ with $i=1,2,3$.

One main aim of this paper is to find the total stationary probability, $P\left( n \right)={{P}_{0}}\left( n \right)+{{P}_{1}}\left( n \right)$, based on \eqref{eq:3}. We point out that stationary distributions have been derived if ${{\tilde{K}}_{i}}\left( n \right)$ ($i=1,2,3$) are all constants \cite{s29,s42,s43}. However, if the normalized ${{\tilde{K}}_{i}}\left( n \right)$ are nonlinear functions of $n$, it seems to us that the analytical expression of steady-state protein distribution has not been derived from \eqref{eq:3} so far. In fact, if the form of ${{\tilde{K}}_{i}}\left( n \right)$ is general, directly solving \eqref{eq:3} is very difficult. We will develop a technique (in fact an analytical framework) to derive the formal expression of stationary protein distribution in a general case (i.e., ${{\tilde{K}}_{i}}\left( n \right)$ with $1\le i\le 3$ are arbitrary functions of $n$).

\section{The exact solution to the CME}
In order to derive the formal expression of stationary protein distribution, our basic idea is that we first take $P\left( 0 \right)$ and ${{P}_{0}}\left( 0 \right)$ as two parameters, then show that $P\left( n \right)$ and ${{P}_{1}}\left( n \right)$ can be formally expressed as the linear combinations of $P\left( 0 \right)$ and ${{P}_{0}}\left( 0 \right)$, and finally give the formal expressions of $P\left( 0 \right)$ and ${{P}_{0}}\left( 0 \right)$ according to the probability conservative condition.

For clarity, we establish the following theorem:
\begin{theorem}\label{1th}
The solution to \eqref{eq:3} can be formally expressed as
\begin{equation}\label{eq:4}
P\left( n \right)=\frac{1}{n!}\frac{{{a}_{n}}-C{{b}_{n}}}{1+\sum\nolimits_{i=1}^{\infty }{{\left( {{a}_{i}}-C{{b}_{i}} \right)}/{i!}\;}}, n=1,2,\cdots,
\end{equation}
where ${{a}_{1}}={{b}_{1}}={{\tilde{K}}_{3}}\left( 0 \right)$, ${{a}_{2}}={{\tilde{K}}_{3}}\left( 1 \right)\left( {{{\tilde{K}}}_{3}}\left( 0 \right)+{{{\tilde{K}}}_{2}}\left( 0 \right) \right)$, ${{b}_{2}}={{\tilde{K}}_{3}}\left( 1 \right)\left( {{{\tilde{K}}}_{3}}\left( 0 \right)+{{{\tilde{K}}}_{2}}\left( 0 \right)+{{{\tilde{K}}}_{1}}\left( 0 \right) \right)$, and for $n\ge 2$, we have
\begin{equation}\label{eq:4a}
\begin{aligned}
{{a}_{n+1}}=&\frac{{{{\tilde{K}}}_{3}}\left( n \right)\left( n-1+{{{\tilde{K}}}_{1}}\left( n-1 \right)+{{{\tilde{K}}}_{2}}\left( n-1 \right)+{{{\tilde{K}}}_{3}}\left( n-1 \right) \right)}{{{{\tilde{K}}}_{3}}\left( n-1 \right)}{{a}_{n}}\\
&-{{\tilde{K}}_{3}}\left( n \right)\left( n-1+{{{\tilde{K}}}_{1}}\left( n-1 \right) \right){{a}_{n-1}},
\end{aligned}
\end{equation}
\begin{equation}\label{eq:4b}
\begin{aligned}
{{b}_{n+1}}=&\frac{{{{\tilde{K}}}_{3}}\left( n \right)\left( n-1+{{{\tilde{K}}}_{1}}\left( n-1 \right)+{{{\tilde{K}}}_{2}}\left( n-1 \right)+{{{\tilde{K}}}_{3}}\left( n-1 \right) \right)}{{{{\tilde{K}}}_{3}}\left( n-1 \right)}{{b}_{n}}\\
&-{{\tilde{K}}_{3}}\left( n \right)\left( n-1+{{{\tilde{K}}}_{1}}\left( n-1 \right) \right){{b}_{n-1}}.
\end{aligned}
\end{equation}
In \eqref{eq:4},
\begin{equation}\label{eq:4c}
C=\underset{N\to \infty }{\mathop{\lim }}\,\frac{{{{\tilde{K}}}_{2}}\left( 0 \right)+\sum\nolimits_{i=1}^{N}{{\left[ \left( {{a}_{i}}+{{c}_{i}} \right){{{\tilde{K}}}_{2}}\left( i \right)+{{c}_{i}}{{{\tilde{K}}}_{1}}\left( i \right) \right]}/{i!}\;}}{{{{\tilde{K}}}_{1}}\left( 0 \right)+{{{\tilde{K}}}_{2}}\left( 0 \right)+\sum\nolimits_{i=1}^{N}{{\left[ \left( {{b}_{i}}+{{d}_{i}} \right){{{\tilde{K}}}_{2}}\left( i \right)+{{d}_{i}}{{{\tilde{K}}}_{1}}\left( i \right) \right]}/{i!}\;}},
\end{equation}
where
\begin{equation}\label{eq:4d}
{{c}_{n}}=\sum\limits_{i=1}^{n-1}{{{{\tilde{K}}}_{2}}\left( i \right){{a}_{i}}\prod\limits_{j=i+1}^{n-1}{\left( j+{{{\tilde{K}}}_{1}}\left( j \right)+{{{\tilde{K}}}_{2}}\left( j \right) \right)}}+{{\tilde{K}}_{2}}\left( 0 \right)\prod\limits_{i=1}^{n-1}{\left( i+{{{\tilde{K}}}_{1}}\left( i \right)+{{{\tilde{K}}}_{2}}\left( i \right) \right)},
\end{equation}
\begin{equation}\label{eq:4e}
{{d}_{n}}=\sum\limits_{i=1}^{n-1}{{{{\tilde{K}}}_{2}}\left( i \right){{b}_{i}}\prod\limits_{j=i+1}^{n-1}{\left( j+{{{\tilde{K}}}_{1}}\left( j \right)+{{{\tilde{K}}}_{2}}\left( j \right) \right)}}+\prod\limits_{i=0}^{n-1}{\left( i+{{{\tilde{K}}}_{1}}\left( i \right)+{{{\tilde{K}}}_{2}}\left( i \right) \right)}.
\end{equation}
\end{theorem}

In order to prove this theorem, we first sum up two equations in \eqref{eq:3}. This will yield
$$
\left( m+1 \right)P\left( m+1 \right)-mP\left( m \right)={{\tilde{K}}_{3}}\left( m \right){{P}_{1}}\left( m \right)-{{\tilde{K}}_{3}}\left( m-1 \right){{P}_{1}}\left( m-1 \right),
$$
where $m=0,1,2,\cdots $ and we define ${{\tilde{K}}_{3}}\left( -1 \right)=0$. Furthermore, summing up both sides of this equation over $m$ from $m=0$ to $m=n$ yields the following relationship
\begin{equation}\label{eq:5}
P\left( n+1 \right)=\frac{{{{\tilde{K}}}_{3}}\left( n \right)}{n+1}{{P}_{1}}\left( n \right),
\end{equation}
where $P\left( n \right)={{P}_{0}}\left( n \right)+{{P}_{1}}\left( n \right)$ and $n=0,1,2,\cdots $. Then, by substituting ${{P}_{1}}\left( n \right)=P\left( n \right)-{{P}_{0}}\left( n \right)$ into the first equation of \eqref{eq:3}, we have
$$
-{{\tilde{K}}_{1}}\left( n \right){{P}_{0}}\left( n \right)+{{\tilde{K}}_{2}}\left( n \right)\left[ P\left( n \right)-{{P}_{0}}\left( n \right) \right]+\left( n+1 \right){{P}_{0}}\left( n+1 \right)-n{{P}_{0}}\left( n \right)=0,
$$
which can be rewritten as
$$
{{P}_{0}}\left( n+1 \right)=\frac{{{{\tilde{K}}}_{1}}\left( n \right)+{{{\tilde{K}}}_{2}}\left( n \right)+n}{n+1}{{P}_{0}}\left( n \right)-\frac{{{{\tilde{K}}}_{2}}\left( n \right)}{n+1}P\left( n \right),
$$
or
\begin{equation}\label{eq:6}
{{P}_{0}}\left( n \right)=\frac{n-1+{{{\tilde{K}}}_{1}}\left( n-1 \right)\text{+}{{{\tilde{K}}}_{2}}\left( n-1 \right)}{n}{{P}_{0}}\left( n-1 \right)-\frac{{{{\tilde{K}}}_{2}}\left( n-1 \right)}{n}P\left( n-1 \right),
\end{equation}
where $n=1,2,\cdots $. By the mathematical induction, we can easily prove the following lemma.

\begin{lemma}\label{lm:1}
If ${{x}_{n}}={{a}_{n}}{{x}_{n-1}}+{{b}_{n}}$, where $n=1,2,\cdots $, then ${{x}_{n}}={{x}_{0}}\prod\limits_{i=1}^{n}{{{a}_{i}}}+\sum\limits_{i=1}^{n-1}{{{b}_{i}}\prod\limits_{j=i}^{n-1}{{{a}_{j+1}}}+{{b}_{n}}}$.
\end{lemma}

When this lemma is applied to \eqref{eq:5}, ${{P}_{0}}\left( n \right)$ can be expressed as
\begin{equation}\label{eq:7}
\begin{aligned}
{{P}_{0}}\left( n \right)=&\frac{1}{n!}{{P}_{0}}\left( 0 \right)\prod\limits_{i=0}^{n-1}{\left( i+{{{\tilde{K}}}_{1}}\left( i \right)+{{{\tilde{K}}}_{2}}\left( i \right) \right)}\\
&-\sum\limits_{i=0}^{n-2}{\frac{{{{\tilde{K}}}_{2}}\left( i \right)}{i+1}P\left( i \right)\prod\limits_{j=i+1}^{n-1}{\frac{j+{{{\tilde{K}}}_{1}}\left( j \right)+{{{\tilde{K}}}_{2}}\left( j \right)}{j+1}}}-\frac{{{{\tilde{K}}}_{2}}\left( n-1 \right)}{n}P\left( n-1 \right).
\end{aligned}
\end{equation}
Thus,
$$
\begin{aligned}
  {{P}_{1}}\left( n \right)=&P\left( n \right)-{{P}_{0}}\left( n \right)\\
  =&P\left( n \right)-\frac{1}{n!}{{P}_{0}}\left( 0 \right)\prod\limits_{i=0}^{n-1}{\left( i+{{{\tilde{K}}}_{1}}\left( i \right)+{{{\tilde{K}}}_{2}}\left( i \right) \right)}\\
  &+\sum\limits_{i=0}^{n-2}{\frac{{{{\tilde{K}}}_{2}}\left( i \right)}{i+1}P\left( i \right)\prod\limits_{j=i+1}^{n-1}{\frac{j+{{{\tilde{K}}}_{1}}\left( j \right)+{{{\tilde{K}}}_{2}}\left( j \right)}{j+1}}}-\frac{{{{\tilde{K}}}_{2}}\left( n-1 \right)}{n}P\left( n-1 \right).
\end{aligned}
$$
Substituting it into \eqref{eq:5} yields
\begin{equation}\label{eq:8}
\begin{aligned}
P\left( n\text{+}1 \right)=&\frac{{{{\tilde{K}}}_{3}}\left( n \right)}{n+1}P\left( n \right)+\frac{{{{\tilde{K}}}_{3}}\left( n \right)}{n+1}\sum\limits_{i=0}^{n-1}{\frac{{{{\tilde{K}}}_{2}}\left( i \right)}{i+1}P\left( i \right)\prod\limits_{j=i+1}^{n-1}{\frac{j+{{{\tilde{K}}}_{1}}\left( j \right)\text{+}{{{\tilde{K}}}_{2}}\left( j \right)}{j+1}}}\\
&-\frac{{{{\tilde{K}}}_{3}}\left( n \right)}{\left( n+1 \right)!}{{P}_{0}}\left( 0 \right)\prod\limits_{i=0}^{n-1}{\left( i+{{{\tilde{K}}}_{1}}\left( i \right)+{{{\tilde{K}}}_{2}}\left( i \right) \right)},
\end{aligned}
\end{equation}
where $n=1,2,\cdots $. Note that \eqref{eq:8} is an iterative system, so it is easily solved. In the following, we will take $P\left( 0 \right)$ and ${{P}_{0}}\left( 0 \right)$ as two parameters, which will be determined later. By the mathematical induction again, we can prove the following lemma, which is a main result of this paper.

\begin{lemma}\label{lm:2}
Stationary protein distribution $P\left( n \right)$ can be formally expressed as
\begin{equation}\label{eq:9}
P\left( n \right)=\frac{1}{n!}\left[ {{a}_{n}}P\left( 0 \right)-{{b}_{n}}{{P}_{0}}\left( 0 \right) \right], n=1,2,\cdots,
\end{equation}
where ${{a}_{1}}={{b}_{1}}={{\tilde{K}}_{3}}\left( 0 \right)$, ${{a}_{n}}$ and ${{b}_{n}}$ with $n\ge 2$ are determined according to the following formula respectively:
\begin{equation}\label{eq:9a}
\begin{aligned}
{{a}_{n}}=&{{{\tilde{K}}}_{3}}\left( n-1 \right){{a}_{n-1}}+{{{\tilde{K}}}_{3}}\left( n-1 \right)\sum\limits_{i=1}^{n-2}{{{{\tilde{K}}}_{2}}\left( i \right){{a}_{i}}\prod\limits_{j=i+1}^{n-2}{\left( j+{{{\tilde{K}}}_{1}}\left( j \right)+{{{\tilde{K}}}_{2}}\left( j \right) \right)}}\\
&+{{{\tilde{K}}}_{3}}\left( n-1 \right){{{\tilde{K}}}_{2}}\left( 0 \right)\prod\limits_{i=1}^{n-2}{\left( i+{{{\tilde{K}}}_{1}}\left( i \right)\text{+}{{{\tilde{K}}}_{2}}\left( i \right) \right)},
\end{aligned}
\end{equation}
\begin{equation}\label{eq:9b}
\begin{aligned}
{{b}_{n}}=&{{{\tilde{K}}}_{3}}\left( n-1 \right){{b}_{n-1}}+{{{\tilde{K}}}_{3}}\left( n-1 \right)\sum\limits_{i=1}^{n-2}{{{{\tilde{K}}}_{2}}\left( i \right){{b}_{i}}\prod\limits_{j=i+1}^{n-2}{\left( j+{{{\tilde{K}}}_{1}}\left( j \right)+{{{\tilde{K}}}_{2}}\left( j \right) \right)}} \\
&+{{{\tilde{K}}}_{3}}\left( n-1 \right)\prod\limits_{i=0}^{n-2}{\left( i+{{{\tilde{K}}}_{1}}\left( i \right)+{{{\tilde{K}}}_{2}}\left( i \right) \right)}.
\end{aligned}
\end{equation}
\end{lemma}

\begin{proof}
By \eqref{eq:5}, we have $P\left( 1 \right)={{\tilde{K}}_{3}}\left( 0 \right){{P}_{1}}\left( 0 \right)={{\tilde{K}}_{3}}\left( 0 \right)\left[ P\left( 0 \right)-{{P}_{0}}\left( 0 \right) \right]$, implying that ${{a}_{1}}={{b}_{1}}={{\tilde{K}}_{3}}\left( 0 \right)$. Therefore, \eqref{eq:9} holds. Assume that \eqref{eq:9}, \eqref{eq:9a} and \eqref{eq:9b} hold for $n\le k$. Now, consider the case of $n=k+1$. In this case, it follows from \eqref{eq:8} that
$$
\begin{aligned}
P\left( k\text{+}1 \right)=&\frac{{{{\tilde{K}}}_{3}}\left( k \right)}{k+1}P\left( k \right)-\frac{{{{\tilde{K}}}_{3}}\left( k \right)}{\left( k+1 \right)!}{{P}_{0}}\left( 0 \right)\prod\limits_{i=0}^{k-1}{\left( i+{{{\tilde{K}}}_{1}}\left( i \right)+{{{\tilde{K}}}_{2}}\left( i \right) \right)} \\
&+\frac{{{{\tilde{K}}}_{3}}\left( k \right)}{\left( k+1 \right)!}{{{\tilde{K}}}_{2}}\left( 0 \right)P\left( 0 \right)\prod\limits_{i=1}^{k-1}{\left( i+{{{\tilde{K}}}_{1}}\left( i \right)+{{{\tilde{K}}}_{2}}\left( i \right) \right)}\\
&+\frac{{{{\tilde{K}}}_{3}}\left( k \right)}{k+1}\sum\limits_{i=1}^{k-1}{\frac{{{{\tilde{K}}}_{2}}\left( i \right)}{i+1}P\left( i \right)\prod\limits_{j=i+1}^{k-1}{\frac{j+{{{\tilde{K}}}_{1}}\left( j \right)\text{+}{{{\tilde{K}}}_{2}}\left( j \right)}{j+1}}}.
\end{aligned}
$$
By the induction hypothesis, we have
$$
\begin{aligned}
P\left( k\text{+}1 \right)=&\frac{{{{\tilde{K}}}_{3}}\left( k \right)}{\left( k+1 \right)!}\left( {{a}_{k}}P\left( 0 \right)-{{b}_{k}}{{P}_{0}}\left( 0 \right) \right)-\frac{{{{\tilde{K}}}_{3}}\left( k \right)}{\left( k+1 \right)!}{{P}_{0}}\left( 0 \right)\prod\limits_{i=0}^{k-1}{\left( i+{{{\tilde{K}}}_{1}}\left( i \right)+{{{\tilde{K}}}_{2}}\left( i \right) \right)}\\
&+\frac{{{{\tilde{K}}}_{3}}\left( k \right)}{\left( k+1 \right)!}{{{\tilde{K}}}_{2}}\left( 0 \right)P\left( 0 \right)\prod\limits_{i=1}^{k-1}{\left( i+{{{\tilde{K}}}_{1}}\left( i \right)\text{+}{{{\tilde{K}}}_{2}}\left( i \right) \right)} \\ &+\frac{{{{\tilde{K}}}_{3}}\left( k \right)}{\left( k+1 \right)!}\sum\limits_{i=1}^{k-1}{{{{\tilde{K}}}_{2}}\left( i \right)\left( {{a}_{i}}P\left( 0 \right)-{{b}_{i}}{{P}_{0}}\left( 0 \right) \right)\prod\limits_{j=i+1}^{k-1}{\left( j+{{{\tilde{K}}}_{1}}\left( j \right)\text{+}{{{\tilde{K}}}_{2}}\left( j \right) \right)}}.
\end{aligned}
$$
Merging the terms for $P\left( 0 \right)$ and ${{P}_{0}}\left( 0 \right)$, we have
$$
P\left( k\text{+}1 \right)=\frac{1}{\left( k+1 \right)!}\left( {{a}_{k+1}}P\left( 0 \right)-{{b}_{k+1}}{{P}_{0}}\left( 0 \right) \right),
$$
where
$$
\begin{aligned}
{{a}_{k+1}}=&{{{\tilde{K}}}_{3}}\left( k \right){{a}_{k}}+{{{\tilde{K}}}_{3}}\left( k \right)\sum\limits_{i=1}^{k-1}{{{{\tilde{K}}}_{2}}\left( i \right){{a}_{i}}\prod\limits_{j=i+1}^{k-1}{\left( j+{{{\tilde{K}}}_{1}}\left( j \right)\text{+}{{{\tilde{K}}}_{2}}\left( j \right) \right)}}\\
&+{{{\tilde{K}}}_{3}}\left( k \right){{{\tilde{K}}}_{2}}\left( 0 \right)\prod\limits_{i=1}^{k-1}{\left( i+{{{\tilde{K}}}_{1}}\left( i \right)\text{+}{{{\tilde{K}}}_{2}}\left( i \right) \right)},
\end{aligned}
$$
$$
\begin{aligned}
 {{b}_{k+1}}=&{{{\tilde{K}}}_{3}}\left( k \right){{b}_{k}}+{{{\tilde{K}}}_{3}}\left( k \right)\sum\limits_{i=1}^{k-1}{{{{\tilde{K}}}_{2}}\left( i \right){{b}_{i}}\prod\limits_{j=i+1}^{k-1}{\left( j+{{{\tilde{K}}}_{1}}\left( j \right)\text{+}{{{\tilde{K}}}_{2}}\left( j \right) \right)}}\\
 &+{{{\tilde{K}}}_{3}}\left( k \right)\prod\limits_{i=0}^{k-1}{\left( i+{{{\tilde{K}}}_{1}}\left( i \right)+{{{\tilde{K}}}_{2}}\left( i \right) \right)}.
 \end{aligned}
$$
This implies that \eqref{eq:9} with \eqref{eq:9a} and \eqref{eq:9b} holds for $n=k+1$. According to the mathematical induction, \eqref{eq:9} with \eqref{eq:9a} and \eqref{eq:9b} holds for all $n\ge 1$. Lemma \ref{lm:2} is thus proven.
\end{proof}

Lemma \ref{lm:2} indicates that all ${{a}_{n}}$ and ${{b}_{n}}$ can iteratively be calculated. Therefore, this lemma actually provides a method for calculating the stationary probability distribution in an ON-OFF model of gene expression with general feedback regulations. Note that both ${{a}_{n}}$ and ${{b}_{n}}$ are positive for all $n$, and are monotonically increasing functions of $n$.

Substituting \eqref{eq:9} with \eqref{eq:9a} and \eqref{eq:9b} into \eqref{eq:6}, we have
$$
\begin{aligned}
 {{P}_{0}}\left( n \right)=&\frac{1}{n!}{{P}_{0}}\left( 0 \right)\prod\limits_{i=0}^{n-1}{\left( i+{{{\tilde{K}}}_{1}}\left( i \right)+{{{\tilde{K}}}_{2}}\left( i \right) \right)}
 -\frac{{{{\tilde{K}}}_{2}}\left( 0 \right)P\left( 0 \right)}{n!}\prod\limits_{i=1}^{n-1}{\left( i+{{{\tilde{K}}}_{1}}\left( i \right)+{{{\tilde{K}}}_{2}}\left( i \right) \right)}\\
&-\sum\limits_{i=1}^{n-1}{\frac{{{{\tilde{K}}}_{2}}\left( i \right)}{i+1}P\left( i \right)\prod\limits_{j=i+1}^{n-1}{\frac{j+{{{\tilde{K}}}_{1}}\left( j \right)+{{{\tilde{K}}}_{2}}\left( j \right)}{j+1}}}.
 \end{aligned}
$$
Using $P\left( i \right)=\left[ {{a}_{i}}P\left( 0 \right)-{{b}_{i}}{{P}_{0}}\left( i \right) \right]/{i!}$, we further have
$$
\begin{aligned}
   {{P}_{0}}\left( n \right)=&\frac{1}{n!}{{P}_{0}}\left( 0 \right)\prod\limits_{i=0}^{n-1}{\left( i+{{{\tilde{K}}}_{1}}\left( i \right)+{{{\tilde{K}}}_{2}}\left( i \right) \right)}-\frac{{{{\tilde{K}}}_{2}}\left( 0 \right)P\left( 0 \right)}{n!}\prod\limits_{i=1}^{n-1}{\left( i+{{{\tilde{K}}}_{1}}\left( i \right)+{{{\tilde{K}}}_{2}}\left( i \right) \right)}\\
&-\sum\limits_{i=1}^{n-1}{\frac{{{{\tilde{K}}}_{2}}\left( i \right)}{i+1}\frac{{{a}_{i}}P\left( 0 \right)-{{b}_{i}}{{P}_{0}}\left( i \right)}{i!}\prod\limits_{j=i+1}^{n-1}{\frac{j+{{{\tilde{K}}}_{1}}\left( j \right)+{{{\tilde{K}}}_{2}}\left( j \right)}{j+1}}}\\
=&-\frac{1}{n!}\left[ \sum\limits_{i=1}^{n-1}{{{{\tilde{K}}}_{2}}\left( i \right){{a}_{i}}\prod\limits_{j=i+1}^{n-1}{\left( j+{{{\tilde{K}}}_{1}}\left( j \right)+{{{\tilde{K}}}_{2}}\left( j \right) \right)}}+{{{\tilde{K}}}_{2}}\left( 0 \right)\prod\limits_{i=1}^{n1}{\left( i+{{{\tilde{K}}}_{1}}\left( i \right)+{{{\tilde{K}}}_{2}}\left( i \right) \right)} \right]P\left( 0 \right)\\
&+\frac{1}{n!}\left[ \sum\limits_{i=1}^{n-1}{{{{\tilde{K}}}_{2}}\left( i \right){{b}_{i}}\prod\limits_{j=i+1}^{n-1}{\left( j+{{{\tilde{K}}}_{1}}\left( j \right)+{{{\tilde{K}}}_{2}}\left( j \right) \right)}}+\prod\limits_{i=0}^{n-1}{\left( i+{{{\tilde{K}}}_{1}}\left( i \right)+{{{\tilde{K}}}_{2}}\left( i \right) \right)} \right]{{P}_{0}}\left( 0 \right).
 \end{aligned}
$$
Therefore, ${{P}_{0}}\left( n \right)$ can be formally expressed as
\begin{equation}\label{eq:10}
 {{P}_{0}}\left( n \right)=\frac{1}{n!}\left[ -{{c}_{n}}P\left( 0 \right)+{{d}_{n}}{{P}_{0}}\left( 0 \right) \right],
\end{equation}
where $n\ge 1$, and ${{c}_{n}}$ and ${{d}_{n}}$ are calculated according to \eqref{eq:4d} and \eqref{eq:4e} respectively. Because of ${{P}_{1}}\left( n \right)=P\left( n \right)-{{P}_{0}}\left( n \right)$, $ {{P}_{1}}\left( n \right)$ can be formally expressed as
\begin{equation}\label{eq:11}
{{P}_{1}}\left( n \right)=\frac{1}{n!}\left[ \left( {{a}_{n}}+{{c}_{n}} \right)P\left( 0 \right)-\left( {{b}_{n}}+{{d}_{n}} \right){{P}_{0}}\left( 0 \right) \right],
\end{equation}
where $n=1,2,\cdots $. In the Appendix, we have simplified \eqref{eq:9a} and \eqref{eq:9b} as \eqref{eq:4a} and \eqref{eq:4b}, respectively.

Now, we only need to determine $P\left( 0 \right)$ and ${{P}_{0}}\left( 0 \right)$. First, since we have assumed that the stationary protein distribution exists, this implies that the series $\sum\nolimits_{n=1}^{\infty }{{\left[ {{a}_{n}}P\left( 0 \right)-{{b}_{n}}{{P}_{0}}\left( 0 \right) \right]}/{n!}\;}$ converges due to the probability conservative condition given by $\sum\nolimits_{n=1}^{\infty }{P\left( n \right)}=1$. Besides, both series $\sum\nolimits_{n=1}^{\infty }{\left[ {{d}_{n}}{{P}_{0}}\left( 0 \right)-{{c}_{n}}P\left( 0 \right) \right]/{n!}\;}$ and series $\sum\nolimits_{n=1}^{\infty }{{\left[ \left( {{a}_{n}}+{{c}_{n}} \right)P\left( 0 \right)-\left( {{b}_{n}}+{{d}_{n}} \right){{P}_{0}}\left( 0 \right) \right]}/{n!}\;}$ are also convergent due to $P\left( n \right)={{P}_{0}}\left( n \right)+{{P}_{1}}\left( n \right)$. However, we point out that the single series, $\left\{ {{{a}_{n}}}/{n!}\; \right\}$ or $\left\{ {{{b}_{n}}}/{n!}\; \right\}$ would be divergent. For this, consider special cases: ${{K}_{1}}\left( n \right)=\alpha +f\frac{{{\left( {n}/{{{D}_{1}}}\; \right)}^{{{H}_{1}}}}}{1+{{\left( {n}/{{{D}_{1}}}\; \right)}^{{{H}_{1}}}}}$, ${{K}_{2}}\left( n \right)=\beta +g\frac{{{\left( {n}/{{{D}_{2}}}\; \right)}^{{{H}_{2}}}}}{1+{{\left( {n}/{{{D}_{2}}}\; \right)}^{{{H}_{2}}}}}$, and ${{K}_{3}}\left( n \right)=\mu +\xi \frac{{{\left( {n}/{{{D}_{2}}}\; \right)}^{{{H}_{2}}}}}{1+{{\left( {n}/{{{D}_{2}}}\; \right)}^{{{H}_{2}}}}}$, where $\alpha $ and $\beta $ represent the basal transition rates between ON and OFF states, $f$, $g$ and $\xi $ represent feedback strengths, and ${{D}_{i}}$ are disassociation coefficents for biochemical reactions associated with feedback regulations. Numerical results are demonstrated in Figure~\ref{f1}. Specifically, if ${{\tilde{K}}_{1}}\left( n \right)=\tilde{\alpha }$, ${{\tilde{K}}_{2}}\left( n \right)=\tilde{\beta }$, and ${{\tilde{K}}_{3}}\left( n \right)=\tilde{\mu }$, the two series are all divergent if $\tilde{\alpha }+\tilde{\beta }>2$, converge to a positive number if $\tilde{\alpha }+\tilde{\beta }=2$, and converge to zero if $\tilde{\alpha }+\tilde{\beta }<2$, referring to Figure~\ref{f1}(A,B). If ${{\tilde{K}}_{1}}\left( n \right)=\tilde{\alpha }+\tilde{f}\frac{{{n}^{h}}}{{{D}^{h}}+{{n}^{h}}}$, ${{\tilde{K}}_{2}}\left( n \right)=\tilde{\beta }+\tilde{g}\frac{{{n}^{h}}}{{{D}^{h}}+{{n}^{h}}}$£¬${{\tilde{K}}_{3}}\left( n \right)=\tilde{\mu }$, they are divergent if $\tilde{\alpha }+\tilde{f}+\tilde{\beta }+\tilde{g}>2$, converge to a positive number if $\tilde{\alpha }+\tilde{f}+\tilde{\beta }+\tilde{g}=2$, and converge to zero if $\tilde{\alpha }+\tilde{f}+\tilde{\beta }+\tilde{g}<2$, referring to Figure~\ref{f1}(C,D). On the convergence of $\left\{ {{{a}_{n}}}/{n!}\; \right\}$ or $\left\{ {{{b}_{n}}}/{n!}\; \right\}$, see discussions in Appendix.

\begin{figure}[h!]
\centering
\subfloat{\includegraphics[scale=0.9]{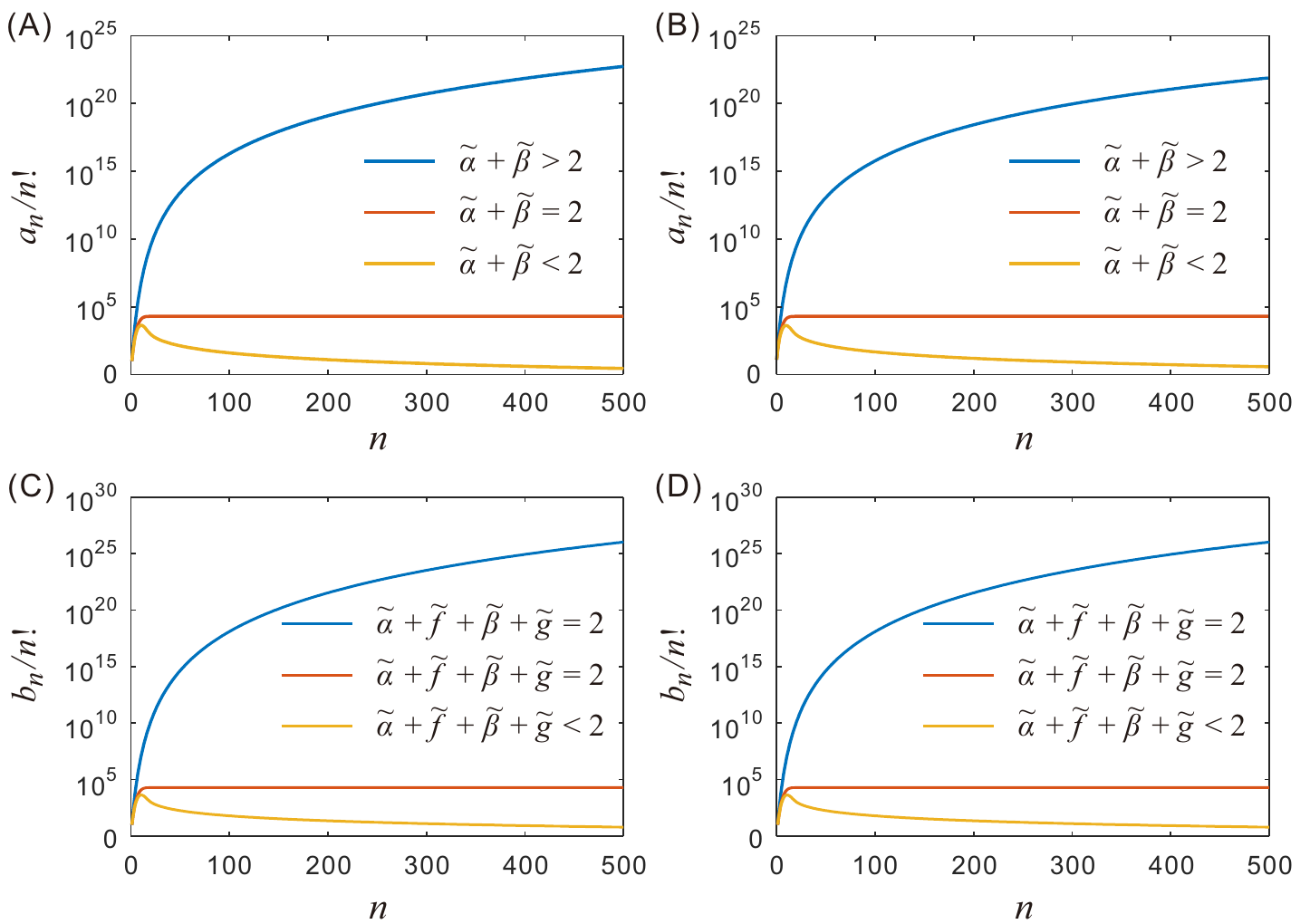}}
\caption{Convergence of series $\left\{ {{{a}_{n}}}/{n!}\; \right\}$ and $\left\{ {{{b}_{n}}}/{n!}\; \right\}$. (A,B) ${{\tilde{K}}_{1}}\left( n \right)=\tilde{\alpha }$, ${{\tilde{K}}_{2}}\left( n \right)=\tilde{\beta }$, and ${{\tilde{K}}_{3}}\left( n \right)=\tilde{\mu }$. We set $\tilde{\alpha }=1$, $\tilde{\beta }=10$, $\tilde{\mu }=\text{1}0$ for $\tilde{\alpha }+\tilde{\beta }>2$; $\tilde{\alpha }=1$, $\tilde{\beta }=1$, $\tilde{\mu }=\text{1}0$ for $\tilde{\alpha }+\tilde{\beta }=2$; and $\tilde{\alpha }=\text{0}.1$, $\tilde{\beta }=\text{0}\text{.5}$, $\tilde{\mu }=\text{1}0$ for $\tilde{\alpha }+\tilde{\beta }<2$. (C,D) ${{\tilde{K}}_{1}}\left( n \right)=\tilde{\alpha }+\tilde{f}\frac{{{n}^{h}}}{{{D}^{h}}+{{n}^{h}}}$, ${{\tilde{K}}_{2}}\left( n \right)=\tilde{\beta }+\tilde{g}\frac{{{n}^{h}}}{{{D}^{h}}+{{n}^{h}}}$£¬${{\tilde{K}}_{3}}\left( n \right)=\tilde{\mu }$. We set $\tilde{\alpha }=1$, $\tilde{\beta }=10$, $\tilde{\mu }=\text{1}0$, $\tilde{f}=1.2$, $\tilde{g}=1$, $D=\sqrt{10}$, $h=2$ for $\tilde{\alpha }+\tilde{f}+\tilde{\beta }+\tilde{g}>2$; $\tilde{\alpha }=\text{0}\text{.5}$, $\tilde{\beta }=\text{0}\text{.5}$, $\tilde{\mu }=\text{1}0$, $\tilde{f}=\text{0}\text{.5}$, $\tilde{g}=\text{0}\text{.5}$, $D=\sqrt{10}$, $h=2$ for $\tilde{\alpha }+\tilde{f}+\tilde{\beta }+\tilde{g}=2$; and $\tilde{\alpha }=\text{0}.1$, $\tilde{\beta }=\text{0}.1$, $\tilde{\mu }=\text{1}0$, $\tilde{f}=\text{0}.2$, $\tilde{g}=\text{0}.2$, $D=\sqrt{10}$, $h=2$ for $\tilde{\alpha }+\tilde{f}+\tilde{\beta }+\tilde{g}<2$.}
\label{f1}
\end{figure}

Next, summing up the first equation of \eqref{eq:3} over $n$ from 0 to $N$ yields
$$
-{{\tilde{K}}_{1}}\left( 0 \right){{P}_{0}}\left( 0 \right)+{{\tilde{K}}_{2}}\left( 0 \right)\left[ P\left( 0 \right)-{{P}_{0}}\left( 0 \right) \right]+\sum\limits_{n=1}^{N}{\left[ {{{\tilde{K}}}_{2}}\left( n \right){{P}_{1}}\left( n \right)-{{{\tilde{K}}}_{1}}\left( n \right){{P}_{0}}\left( n \right) \right]}-{{\tilde{K}}_{3}}\left( N \right){{P}_{0}}\left( N \right)=0,
$$
which holds for any positive integer $N$. Using the formal expressions of ${{P}_{0}}\left( n \right)$ and ${{P}_{1}}\left( n \right)$ given by \eqref{eq:10} and \eqref{eq:11} above, we thus have the following relationship for all $N\ge 1$
$$
\begin{aligned}
&P\left( 0 \right)\left\{ {{{\tilde{K}}}_{2}}\left( 0 \right)+\sum\limits_{n=1}^{N}{\frac{1}{n!}\left[ {{{\tilde{K}}}_{2}}\left( n \right)\left( {{a}_{n}}+{{c}_{n}} \right)+{{{\tilde{K}}}_{1}}\left( n \right){{c}_{n}} \right]} \right\}\\
 &-{{P}_{0}}\left( 0 \right)\left\{ {{{\tilde{K}}}_{1}}\left( 0 \right)+{{{\tilde{K}}}_{2}}\left( 0 \right)+\sum\limits_{n=1}^{N}{\frac{1}{n!}\left[ {{{\tilde{K}}}_{2}}\left( n \right)\left( {{b}_{n}}+{{d}_{n}} \right)+{{{\tilde{K}}}_{1}}\left( n \right){{d}_{n}} \right]} \right\}-{{{\tilde{K}}}_{3}}\left( N \right){{P}_{0}}\left( N \right)=0.
 \end{aligned}
$$
Assume $\underset{N\to \infty }{\mathop{\lim }}\,{{\tilde{K}}_{3}}\left( N \right){{P}_{0}}\left( N \right)=0$. Note that two positive series $\sum\nolimits_{n=1}^{\infty }{{\left[ \left( {{a}_{n}}+{{c}_{n}} \right){{{\tilde{K}}}_{2}}\left( n \right)+{{c}_{n}}{{{\tilde{K}}}_{1}}\left( n \right) \right]}/{n!}\;}$ and $\sum\nolimits_{n=1}^{\infty }{{\left[ \left( {{b}_{n}}+{{d}_{n}} \right){{{\tilde{K}}}_{2}}\left( n \right)+{{d}_{n}}{{{\tilde{K}}}_{1}}\left( n \right) \right]}/{n!}\;}$ are simultaneously convergent or divergent since $P\left( 0 \right)$, ${{P}_{0}}\left( 0 \right)$, and ${{\tilde{K}}_{i}}\left( 0 \right)$ are all finite. If they are convergent, then both $\sum\nolimits_{n=1}^{\infty }{{{{a}_{n}}}/{n!}\;}$ and $\sum\nolimits_{n=1}^{\infty }{{{{c}_{n}}}/{n!}\;}$ are also convergent due to ${{\tilde{\alpha }}_{i}}\le {{\tilde{K}}_{i}}\left( n \right)<{{\tilde{\alpha }}_{i}}+{{\tilde{f}}_{i}}$. Therefore,
\begin{equation}\label{eq:12}
{{P}_{0}}\left( 0 \right)=CP\left( 0 \right),
\end{equation}
where $C$ is given by \eqref{eq:4c}. If they are divergent, then $C$ can still be given via \eqref{eq:4c} (i.e., by summing up the first finite terms in the series). In combination with the probability conservative condition,
$$
1=\sum\limits_{n=0}^{\infty }{P\left( n \right)}=P\left( 0 \right)+\sum\limits_{n=1}^{\infty }{\frac{1}{n!}\left[ {{a}_{n}}P\left( 0 \right)-{{b}_{n}}{{P}_{0}}\left( 0 \right) \right]}=P\left( 0 \right)+P\left( 0 \right)\sum\limits_{n=1}^{\infty }{\frac{1}{n!}\left( {{a}_{n}}-C{{b}_{n}} \right)}.
$$
We can thus determine $P\left( 0 \right)$ and ${{P}_{0}}\left( 0 \right)$, which are given formally by
\begin{equation}\label{eq:13}
P\left( 0 \right)=\underset{N\to \infty }{\mathop{\lim }}\,\frac{1}{1+\sum\nolimits_{i=1}^{N}{{\left( {{a}_{i}}-C{{b}_{i}} \right)}/{i!}\;}}, \quad {{P}_{0}}\left( 0 \right)=\underset{N\to \infty }{\mathop{\lim }}\,\frac{C}{1+\sum\nolimits_{i=1}^{N}{{\left( {{a}_{i}}-C{{b}_{i}} \right)}/{i!}\;}}.
\end{equation}
To that end, the stationary protein distribution can indeed be expressed by \eqref{eq:4}, which is one main result of this paper, where ${{a}_{n}}$ and ${{b}_{n}}$ are determined by \eqref{eq:4a} and \eqref{eq:4b}, and $C$ is given by \eqref{eq:4c}.

In applications, we do not need to calculate ${{a}_{n}}$ and ${{b}_{n}}$ separately. In fact, if we set ${{y}_{n}}={{a}_{n}}P\left( 0 \right)-{{b}_{n}}{{P}_{0}}\left( 0 \right)$ with $n=1,2,\cdots $, then it follows from \eqref{eq:4a} and \eqref{eq:4b} that
\begin{equation}\label{eq:14}
\begin{aligned}
{{y}_{n+1}}=&\frac{{{{\tilde{K}}}_{3}}\left( n \right)\left( n-1+{{{\tilde{K}}}_{1}}\left( n-1 \right)+{{{\tilde{K}}}_{2}}\left( n-1 \right)+{{{\tilde{K}}}_{3}}\left( n-1 \right) \right)}{{{{\tilde{K}}}_{3}}\left( n-1 \right)}{{y}_{n}}\\
&-{{\tilde{K}}_{3}}\left( n \right)\left( n-1+{{{\tilde{K}}}_{1}}\left( n-1 \right) \right){{y}_{n-1}},
\end{aligned}
\end{equation}
where $n=2,3,\cdots $, ${{y}_{0}}=0$ and ${{y}_{1}}={{\tilde{K}}_{3}}\left( 0 \right)\left[ P\left( 0 \right)-{{P}_{0}}\left( 0 \right) \right]$. Note that \eqref{eq:14} is still an iterative system, so ${{y}_{n}}$ can easily be obtained. Also note that ${{y}_{n}}>0$ for all positive integers, $n$.

In a word, the above analysis process gives a framework for calculating stationary protein distributions in an ON-OFF model of gene expression with arbitrary feedback regulations (i.e., ${{K}_{i}}\left( n \right)$ with $i=1,2,3$ are any functions of $n$).

Here we list main steps for calculating the stationary protein distribution:

Step-0. Input parameter values and $N$ (a large positive integer, e.g., $N=200$), and calculate ${{a}_{1}}={{b}_{1}}={{\tilde{K}}_{3}}\left( 0 \right)$, ${{\tilde{K}}_{1}}\left( 0 \right)$ and ${{\tilde{K}}_{2}}\left( 0 \right)$;

Step-1. Set $n=1$;

Step-2. Calculate ${{\tilde{K}}_{i}}\left( n \right)$ ($1\le i\le 3$), ${{a}_{n}}$ and ${{b}_{n}}$ according to \eqref{eq:4a} and \eqref{eq:4b}, as well as ${{c}_{n}}$ and ${{d}_{n}}$ according to \eqref{eq:4d} and \eqref{eq:4e};

Step-3. Update $n+1\to n$. If $n\le N$, then go to Step-2, and turn to the next step (i.e., Step-4) elsewhere;

Step-4. Calculate $C$ according to \eqref{eq:4c}, and $P\left( n \right)$ according to \eqref{eq:4}, where $n=0,1,2,\cdots ,N$;

Step-5. Output $P\left( n \right)$.

\section{Analytical protein distributions in special cases}
In this section, we will reproduce known distributions in three special cases. First, consider the case of ${{K}_{1}}\left( n \right)=\alpha $, ${{K}_{2}}\left( n \right)=\beta $ and ${{K}_{3}}\left( n \right)=\mu $ for all $n$, where $\alpha $, $\beta $ and $\mu $ are positive constants. In this case, the corresponding gene model reduces to the common On-OFF model. For convenience, we denote $\tilde{\alpha }={\alpha }/{\delta }$, $\tilde{\beta }={\beta }/{\delta }$, $\tilde{\mu }={\mu }/{\delta }$. Then, ${{\tilde{K}}_{1}}\left( n \right)=\tilde{\alpha }$, ${{\tilde{K}}_{2}}\left( n \right)=\tilde{\beta }$, ${{\tilde{K}}_{3}}\left( n \right)=\tilde{\mu }$, where $n=0,1,2,\cdots $. Moreover, \eqref{eq:9a} reduces
\begin{equation}\label{eq:15}
{{a}_{n}}=\tilde{\mu }{{a}_{n-1}}+\tilde{\mu }\tilde{\beta }\sum\limits_{i=1}^{n-2}{{{a}_{i}}\prod\limits_{j=i+1}^{n-2}{\left( j+\tilde{\alpha }+\tilde{\beta } \right)}}+\tilde{\mu }\tilde{\beta }\prod\limits_{i=1}^{n-2}{\left( i+\tilde{\alpha }+\tilde{\beta } \right)},
\end{equation}
where ${{a}_{1}}=\tilde{\mu }$, $n=2,3,\cdots $. From \eqref{eq:15}, we can obtain the expressions of all ${{a}_{n}}$, e.g., the initial several ${{a}_{n}}$ are
$$
{{a}_{1}}=\tilde{\mu },\quad {{a}_{2}}=\tilde{\mu }{{a}_{1}}+\tilde{\mu }\tilde{\beta }={{\tilde{\mu }}^{2}}+\tilde{\mu }\tilde{\beta },
$$
$$
{{a}_{3}}={{\tilde{\mu }}^{3}}+2{{\tilde{\mu }}^{2}}\tilde{\beta }+\tilde{\mu }\tilde{\beta }\left( 1+\tilde{\alpha }+\tilde{\beta } \right),
$$
$$
{{a}_{4}}={{\tilde{\mu }}^{4}}+3{{\tilde{\mu }}^{3}}\tilde{\beta }+{{\tilde{\mu }}^{2}}\tilde{\beta }\left( 3+2\tilde{\alpha }+3\tilde{\beta } \right)+\tilde{\mu }\tilde{\beta }\prod\limits_{i=1}^{2}{\left( i+\tilde{\alpha }+\tilde{\beta } \right)},
$$
$$
{{a}_{5}}={{\tilde{\mu }}^{5}}+4{{\tilde{\mu }}^{4}}\tilde{\beta }+3{{\tilde{\mu }}^{3}}\tilde{\beta }\left( 2+\tilde{\alpha }+2\tilde{\beta } \right)+2{{\tilde{\mu }}^{2}}\tilde{\beta }{{\left( 2+\tilde{\alpha }+\tilde{\beta } \right)}^{2}}+\tilde{\mu }\tilde{\beta }\prod\limits_{i=1}^{3}{\left( i+\tilde{\alpha }+\tilde{\beta } \right)}.
$$
Similarly, we can give the expressions of initial several ${{b}_{n}}$, ${{c}_{n}}$ and ${{d}_{n}}$ according to \eqref{eq:4d},\eqref{eq:4e} and \eqref{eq:9b}, respectively. Interestingly, we find, by calculation,
$$
1+\sum\limits_{n=1}^{\infty }{\frac{1}{n!}\frac{\tilde{\alpha }{{c}_{n}}+\tilde{\beta }\left( {{a}_{n}}+{{c}_{n}} \right)}{{\tilde{\beta }}}}={}_{1}{{F}_{1}}\left( \tilde{\alpha },\tilde{\alpha }+\tilde{\beta }+1;-\tilde{\mu } \right),
$$
$$
1+\sum\limits_{n=1}^{\infty }{\frac{1}{n!}\frac{\tilde{\alpha }{{d}_{n}}+\tilde{\beta }\left( {{b}_{n}}+{{d}_{n}} \right)}{\tilde{\alpha }+\tilde{\beta }}}={}_{1}{{F}_{1}}\left( \tilde{\alpha },\tilde{\alpha }+\tilde{\beta };-\tilde{\mu } \right),
$$
where ${}_{1}{{F}_{1}}\left( a,b;z \right)=\sum\limits_{n=0}^{\infty }{\frac{{{\left( a \right)}_{n}}}{{{\left( b \right)}_{n}}}\frac{{{z}^{n}}}{n!}}$ is a hypergeometric function and ${{\left( c \right)}_{n}}$ (the Pochhammer symbol) is defined as ${{\left( c \right)}_{n}}={\Gamma \left( n+c \right)}/{\Gamma \left( c \right)}\;$ with $\Gamma \left( \cdot  \right)$ being the common Gamma function. According to \eqref{eq:4c}, we thus obtain
\begin{equation}\label{eq:16}
C=\frac{{\tilde{\beta }}}{\tilde{\alpha }+\tilde{\beta }}\frac{{}_{1}{{F}_{1}}\left( \tilde{\alpha },\tilde{\alpha }+\tilde{\beta }+1;-\tilde{\mu } \right)}{{}_{1}{{F}_{1}}\left( \tilde{\alpha },\tilde{\alpha }+\tilde{\beta };-\tilde{\mu } \right)}.
\end{equation}
Furthermore, according to \eqref{eq:12}, we have
\begin{equation}\label{eq:17}
P\left( 0 \right)={}_{1}{{F}_{1}}\left( \tilde{\alpha },\tilde{\alpha }+\tilde{\beta };-\tilde{\mu } \right),\quad {{P}_{0}}\left( 0 \right)=\frac{{\tilde{\beta }}}{\tilde{\alpha }+\tilde{\beta }}{}_{1}{{F}_{1}}\left( \tilde{\alpha },\tilde{\alpha }+\tilde{\beta }+1;-\tilde{\mu } \right).
\end{equation}
Note that \eqref{eq:14} reduces to
\begin{equation}\label{eq:18}
{{y}_{n}}=\tilde{\mu }{{y}_{n-1}}+\tilde{\mu }\tilde{\beta }\sum\limits_{i=1}^{n-2}{{{y}_{i}}\prod\limits_{j=i\text{+}1}^{n-2}{\left( j+\tilde{\alpha }+\tilde{\beta } \right)}}+\tilde{\mu }{{\left( \tilde{\alpha }+\tilde{\beta } \right)}_{n-1}}\tilde{\gamma },
\end{equation}
where $\tilde{\gamma }=\frac{{\tilde{\beta }}}{\tilde{\alpha }+\tilde{\beta }}P\left( 0 \right)-{{P}_{0}}\left( 0 \right)$, ${{y}_{1}}=\tilde{\mu }\left[ P\left( 0 \right)-{{P}_{0}}\left( 0 \right) \right]$, and $n=2,3,\cdots $. By tedious calculations, we find
\begin{equation}\label{eq:19}
{{y}_{n}}={{\tilde{\mu }}^{n}}\frac{{{\left( {\tilde{\alpha }} \right)}_{n}}}{{{\left( \tilde{\alpha }+\tilde{\beta } \right)}_{n}}}{}_{1}{{F}_{1}}\left( \tilde{\alpha }+n,\tilde{\alpha }+\tilde{\beta }+n;-\tilde{\mu } \right).
\end{equation}
Therefore, the stationary protein distribution is given by
\begin{equation}\label{eq:20}
P\left( n \right)=\frac{{{{\tilde{\mu }}}^{n}}}{n!}\frac{{{\left( {\tilde{\alpha }} \right)}_{n}}}{{{\left( \tilde{\alpha }+\tilde{\beta } \right)}_{n}}}{}_{1}{{F}_{1}}\left( \tilde{\alpha }+n,\tilde{\alpha }+\tilde{\beta }+n;-\tilde{\mu } \right).
\end{equation}
The similar stationary distribution was also derived for the common ON-OFF model of gene expression at the transcription level \cite{s29,s42,s43}.

Second, consider the case of ${{K}_{1}}\left( n \right)=\alpha +nf$, ${{K}_{2}}\left( n \right)=\beta $ and ${{K}_{3}}\left( n \right)=\mu $, i.e., consider a gene model with a linear positive feedback, where $f$ represents positive feedback strength. In this case, we can show that the stationary protein distribution is given by
\begin{equation}\label{eq:21}
P\left( n \right)=\frac{P\left( 0 \right)}{\left( n \right)!}{{\left( \frac{{\tilde{\mu }}}{1+\tilde{f}} \right)}^{n}}\frac{{{\left( {{\tilde{\alpha }}}/{\left( 1+\tilde{f} \right)}\; \right)}_{n}}}{{{\left( {\left( \tilde{\alpha }+\tilde{\beta } \right)}/{\left( 1+\tilde{f} \right)}\; \right)}_{n}}}{}_{1}{{F}_{1}}\left( n+\frac{{\tilde{\alpha }}}{1+\tilde{f}},n+\frac{\tilde{\alpha }+\tilde{\beta }}{1+\tilde{f}};-\frac{{\tilde{\mu }}}{1+\tilde{f}} \right),
\end{equation}
with
\begin{equation}\label{eq:21a}
P\left( 0 \right)={{\left[ {}_{1}{{F}_{1}}\left( \frac{{\tilde{\alpha }}}{1+\tilde{f}},\frac{\tilde{\alpha }+\tilde{\beta }}{1+\tilde{f}};\frac{\tilde{f}\tilde{\mu }}{1+\tilde{f}} \right) \right]}^{-1}},
\end{equation}
where $\tilde{\alpha }={\alpha }/{\delta }$, $\tilde{\beta }={\beta }/{\delta }$, $\tilde{\mu }={\mu }/{\delta }$ and $\tilde{f}={f}/{\delta }$. Similarly, if we consider a gene model with a linear negative feedback, i.e., ${{K}_{1}}\left( n \right)=\alpha $, ${{K}_{2}}\left( n \right)=\beta +ng$ and ${{K}_{3}}\left( n \right)=\mu $, where $g$ represents negative feedback strength, then the stationary protein distribution takes the form
\begin{equation}\label{eq:22}
P\left( n \right)=\frac{P\left( 0 \right)}{\left( n \right)!}\frac{{{\left( {\tilde{\alpha }} \right)}_{n}}{{\left[ \frac{{\tilde{\mu }}}{{{\left( 1+\tilde{g} \right)}^{2}}} \right]}^{n}}}{{{\left( \frac{\tilde{\alpha }+\tilde{\beta }}{1+\tilde{g}}\text{+}\frac{\tilde{g}\tilde{\mu }}{{{\left( 1+\tilde{g} \right)}^{2}}} \right)}_{n}}}{}_{1}{{F}_{1}}\left( n+\tilde{\alpha },n+\frac{\tilde{\alpha }+\tilde{\beta }}{1+\tilde{g}}\text{+}\frac{\tilde{g}\tilde{\mu }}{{{\left( 1+\tilde{g} \right)}^{2}}};-\frac{{\tilde{\mu }}}{{{\left( 1+\tilde{g} \right)}^{2}}} \right),
\end{equation}
with
\begin{equation}\label{eq:22a}
P\left( 0 \right)={{\left[ {}_{1}{{F}_{1}}\left( \tilde{\alpha },\frac{\tilde{\alpha }+\tilde{\beta }}{1+\tilde{g}}+\frac{\tilde{\mu }\tilde{g}}{{{\left( 1+\tilde{g} \right)}^{2}}};-\frac{\tilde{\mu }\tilde{g}}{{{\left( 1+\tilde{g} \right)}^{2}}} \right) \right]}^{-1}},
\end{equation}
where $\tilde{\alpha }={\alpha }/{\delta }$, $\tilde{\beta }={\beta }/{\delta }$, $\tilde{\mu }={\mu }/{\delta }$ and $\tilde{g}={g}/{\delta }$. The above two analytical distributions are all known results \cite{s14,s15,s23}. Note that if $\tilde{f}=0$ or $\tilde{g}=0$, then \eqref{eq:21} with \eqref{eq:21a} or \eqref{eq:22} with \eqref{eq:22a} reduces to \eqref{eq:20}.

Regarding the effect of feedback on stationary protein distribution, we plot Figure \ref{f2}, which demonstrates that theoretical results (solid lines) are in accordance with numerical results (empty circles). From this figure, we observe that in the absence of negative feedback regulation (i.e., $\tilde{g}=0$), an appropriate positive feedback strength can induce bimodality, referring to Figure \ref{f2}(A-C). Similarly, in the absence of positive feedback regulation (i.e., $\tilde{f}=0$), an appropriate negative strength can also induce bimodality, referring to Figure \ref{f2}(D-F). In any case, bimodal protein distributions can occur only when two normalized fundamental switching rates $\tilde{\alpha }$ and $\tilde{\beta }$ are small.
\begin{figure}[h!]
\centering
\subfloat{\includegraphics[scale=1]{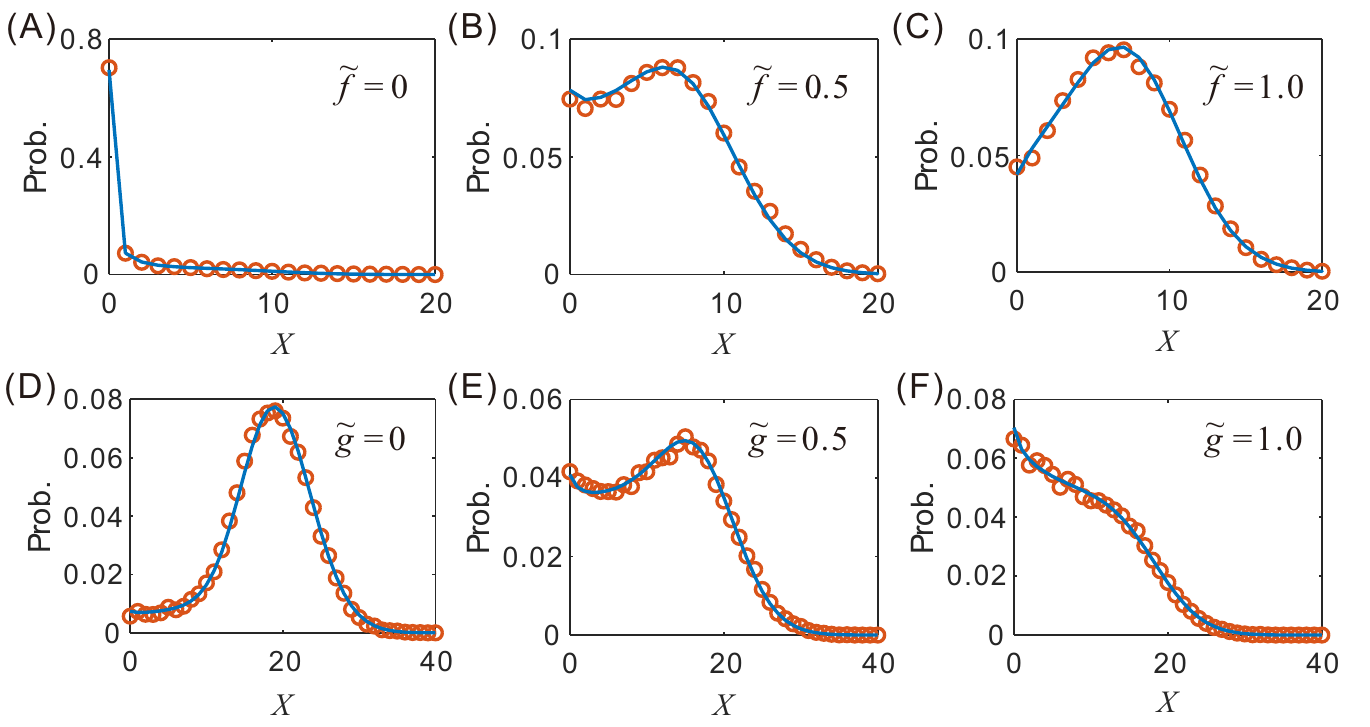}}
\caption{Dependence of steady-state probability distribution on feedback strength, where solid lines correspond to the results obtained by theoretical prediction whereas empty circles to the results obtained by the Gillespie stochastic simulation \cite{s11}. Reaction propensity functions are set as ${{\tilde{K}}_{1}}\left( n \right)=\tilde{\alpha }+\tilde{f}\frac{{{n}^{h}}}{{{D}^{h}}+{{n}^{h}}}$, ${{\tilde{K}}_{2}}\left( n \right)=\tilde{\beta }+\tilde{g}\frac{{{n}^{h}}}{{{D}^{h}}+{{n}^{h}}}$, ${{\tilde{K}}_{3}}\left( n \right)=\tilde{\mu }$. (A-C) The case of $\tilde{g}=0$, where parameter values are set as $\tilde{\alpha }=0.1$, $\tilde{\beta }=\text{0}\text{.6}$, $\tilde{\mu }=\text{10}$, $D=\sqrt{10}$, $h=2$; (D-F) The case of $\tilde{f}=0$, where parameter values are set as $\tilde{\alpha }=0.\text{9}$, $\tilde{\beta }=\text{0}\text{.1}$, $\tilde{\mu }=\text{20}$, $D=\sqrt{10}$, $h=2$.}
\label{f2}
\end{figure}

\section{Conclusion and discussion}
The two-state (or ON-OFF) models have extensively been used in modeling of stochastic gene expression. If feedbacks are not considered or the only linear feedbacks are considered, analytical gene product (mRNA or protein) distributions have been derived. However, the ways of feedback regulation are diverse and the feedbacks are often nonlinear due to the binding of transcription factors to the regulatory sites. If general feedback regulations are characterized by Hill-type functions \cite{s1}, exact analytical distributions of gene products have not been obtained so far. Here, we have developed a general analysis framework to derive the exact protein distribution in a generalized ON-OFF model of stochastic gene expression with arbitrary feedbacks including positive and negative feedbacks as well as posttranscriptional or posttranslational regulation. This technique can easily be extended to modeling and analysis of other similar yet complex biochemical reaction systems.

Although analytical stationary gene product distributions have been derived, sources of stochastic fluctuations in the gene expression levels cannot clearly be seen. In fact, from theses formal distributions, it is difficult to give the explicit decomposition principle for the expression noise. It is also difficult to dissect the contributions of the fractional noisy sources (e.g., the promoter noise, and the noise originating from feedback regulation) to the resulting total noise as done in \cite{s14,s23}. More work or further analysis is needed. In addition, the questions such as how new biological knowledge is discovered from the formal distributions and how design principles in biology are concluded from the formal distributions are worth further investigation.


\section*{Acknowledgements}
This work was supported by grants 11931019, 11775314 and 91530320 from National Natural Science Foundation of China.

\section*{Appendix: On the convergence of the series}


Here we give a simple discussion on the convergence of the series involved in the main text.

First, we simplify \eqref{eq:9a} and \eqref{eq:9b} in the main text. Note that ${{a}_{1}}={{b}_{1}}={{\tilde{K}}_{3}}\left( 0 \right)$. Then, ${{a}_{n}}$ and ${{b}_{n}}$ with $n\ge 2$ can be determined according to the following iterative relationships
\begin{equation}\label{eq:A1a}
\begin{aligned}
   {{a}_{n}}&={{{\tilde{K}}}_{3}}\left( n-1 \right){{a}_{n-1}}+{{{\tilde{K}}}_{3}}\left( n-1 \right)\sum\limits_{i=1}^{n-2}{{{{\tilde{K}}}_{2}}\left( i \right){{a}_{i}}\prod\limits_{j=i+1}^{n-2}{\left( j+{{{\tilde{K}}}_{1}}\left( j \right)+{{{\tilde{K}}}_{2}}\left( j \right) \right)}}\\
  &+{{{\tilde{K}}}_{3}}\left( n-1 \right){{{\tilde{K}}}_{2}}\left( 0 \right)\prod\limits_{i=1}^{n-2}{\left( i+{{{\tilde{K}}}_{1}}\left( i \right)+{{{\tilde{K}}}_{2}}\left( i \right) \right)},
\end{aligned}
\end{equation}
\begin{equation}\label{eq:A1b}
\begin{aligned}
 {{b}_{n}}&={{{\tilde{K}}}_{3}}\left( n-1 \right){{b}_{n-1}}+{{{\tilde{K}}}_{3}}\left( n-1 \right)\sum\limits_{i=1}^{n-2}{{{{\tilde{K}}}_{2}}\left( i \right){{b}_{i}}\prod\limits_{j=i+1}^{n-2}{\left( j+{{{\tilde{K}}}_{1}}\left( j \right)+{{{\tilde{K}}}_{2}}\left( j \right) \right)}} \\
  &+{{{\tilde{K}}}_{3}}\left( n-1 \right)\prod\limits_{i=0}^{n-2}{\left( i+{{{\tilde{K}}}_{1}}\left( i \right)+{{{\tilde{K}}}_{2}}\left( i \right) \right)}.
\end{aligned}
\end{equation}
For $n\ge 3$, we have
$$
\begin{aligned}
&\frac{{{{\tilde{K}}}_{2}}\left( n \right){{a}_{n}}}{\prod\limits_{i=1}^{n}{\left( i+{{{\tilde{K}}}_{1}}\left( i \right)+{{{\tilde{K}}}_{2}}\left( i \right) \right)}}=\frac{{{{\tilde{K}}}_{3}}\left( n-1 \right){{{\tilde{K}}}_{2}}\left( n \right)}{\left( n+{{{\tilde{K}}}_{1}}\left( n \right)+{{{\tilde{K}}}_{2}}\left( n \right) \right){{{\tilde{K}}}_{2}}\left( n-1 \right)}\frac{{{{\tilde{K}}}_{2}}\left( n-1 \right){{a}_{n-1}}}{\prod\limits_{i=1}^{n-1}{\left( i+{{{\tilde{K}}}_{1}}\left( i \right)+{{{\tilde{K}}}_{2}}\left( i \right) \right)}} \\
&\quad\quad+\frac{{{{\tilde{K}}}_{3}}\left( n-1 \right){{{\tilde{K}}}_{2}}\left( n \right)\prod\limits_{i=1}^{n-2}{\left( i+{{{\tilde{K}}}_{1}}\left( i \right)+{{{\tilde{K}}}_{2}}\left( i \right) \right)}}{\prod\limits_{i=1}^{n}{\left( i+{{{\tilde{K}}}_{1}}\left( i \right)+{{{\tilde{K}}}_{2}}\left( i \right) \right)}}\sum\limits_{i=1}^{n-2}{\frac{{{{\tilde{K}}}_{2}}\left( i \right){{a}_{i}}}{\prod\limits_{j=1}^{i}{\left( j+{{{\tilde{K}}}_{1}}\left( j \right)+{{{\tilde{K}}}_{2}}\left( j \right) \right)}}} \\
 &\quad\quad+\frac{{{{\tilde{K}}}_{3}}\left( n-1 \right){{{\tilde{K}}}_{2}}\left( 0 \right){{{\tilde{K}}}_{2}}\left( n \right)\prod\limits_{i=1}^{n-2}{\left( i+{{{\tilde{K}}}_{1}}\left( i \right)+{{{\tilde{K}}}_{2}}\left( i \right) \right)}}{\prod\limits_{i=1}^{n}{\left( i+{{{\tilde{K}}}_{1}}\left( i \right)+{{{\tilde{K}}}_{2}}\left( i \right) \right)}},
   \end{aligned}
$$
where ${{a}_{2}}={{\tilde{K}}_{3}}\left( 1 \right)\left( {{{\tilde{K}}}_{3}}\left( 0 \right)+{{{\tilde{K}}}_{2}}\left( 0 \right) \right)$, ${{b}_{2}}={{\tilde{K}}_{3}}\left( 1 \right)\left( {{{\tilde{K}}}_{3}}\left( 0 \right)+{{{\tilde{K}}}_{2}}\left( 0 \right)+{{{\tilde{K}}}_{1}}\left( 0 \right) \right)$. Furthermore,
$$
\begin{aligned}
&\frac{{{{\tilde{K}}}_{2}}\left( n \right){{a}_{n}}}{\prod\limits_{i=1}^{n}{\left( i+{{{\tilde{K}}}_{1}}\left( i \right)+{{{\tilde{K}}}_{2}}\left( i \right) \right)}}=\frac{{{{\tilde{K}}}_{3}}\left( n-1 \right){{{\tilde{K}}}_{2}}\left( n \right)}{\left( n+{{{\tilde{K}}}_{1}}\left( n \right)+{{{\tilde{K}}}_{2}}\left( n \right) \right){{{\tilde{K}}}_{2}}\left( n-1 \right)}\frac{{{{\tilde{K}}}_{2}}\left( n-1 \right){{a}_{n-1}}}{\prod\limits_{i=1}^{n-1}{\left( i+{{{\tilde{K}}}_{1}}\left( i \right)+{{{\tilde{K}}}_{2}}\left( i \right) \right)}} \\
&\quad+\frac{{{{\tilde{K}}}_{3}}\left( n-1 \right){{{\tilde{K}}}_{2}}\left( n \right)}{\left( n+{{{\tilde{K}}}_{1}}\left( n \right)+{{{\tilde{K}}}_{2}}\left( n \right) \right)\left( n-1+{{{\tilde{K}}}_{1}}\left( n-1 \right)+{{{\tilde{K}}}_{2}}\left( n-1 \right) \right)}\sum\limits_{i=1}^{n-2}{\frac{{{{\tilde{K}}}_{2}}\left( i \right){{a}_{i}}}{\prod\limits_{j=1}^{i}{\left( j+{{{\tilde{K}}}_{1}}\left( j \right)+{{{\tilde{K}}}_{2}}\left( j \right) \right)}}} \\
&\quad+\frac{{{{\tilde{K}}}_{3}}\left( n-1 \right){{{\tilde{K}}}_{2}}\left( n \right){{{\tilde{K}}}_{2}}\left( 0 \right)}{\left( n+{{{\tilde{K}}}_{1}}\left( n \right)+{{{\tilde{K}}}_{2}}\left( n \right) \right)\left( n-1+{{{\tilde{K}}}_{1}}\left( n-1 \right)+{{{\tilde{K}}}_{2}}\left( n-1 \right) \right)}. \\
  \end{aligned}
$$
If we denote ${{A}_{n}}=\frac{{{{\tilde{K}}}_{2}}\left( n \right){{a}_{n}}}{\sum\nolimits_{i=1}^{n}{\left( i+{{{\tilde{K}}}_{1}}\left( i \right)+{{{\tilde{K}}}_{2}}\left( i \right) \right)}}$ and ${{S}_{n}}=\sum\limits_{i=1}^{n}{{{A}_{n}}}$, then
$$
\begin{aligned}
{{A}_{n}}=&\frac{{{{\tilde{K}}}_{3}}\left( n-1 \right){{{\tilde{K}}}_{2}}\left( n \right)}{\left( n+{{{\tilde{K}}}_{1}}\left( n \right)+{{{\tilde{K}}}_{2}}\left( n \right) \right){{{\tilde{K}}}_{2}}\left( n-1 \right)}{{A}_{n-1}}\\
&+\frac{{{{\tilde{K}}}_{3}}\left( n-1 \right){{{\tilde{K}}}_{2}}\left( n \right)}{\left( n+{{{\tilde{K}}}_{1}}\left( n \right)+{{{\tilde{K}}}_{2}}\left( n \right) \right)\left( n-1+{{{\tilde{K}}}_{1}}\left( n-1 \right)+{{{\tilde{K}}}_{2}}\left( n-1 \right) \right)}{{S}_{n-2}} \\
&+\frac{{{{\tilde{K}}}_{3}}\left( n-1 \right){{{\tilde{K}}}_{2}}\left( n \right){{{\tilde{K}}}_{2}}\left( 0 \right)}{\left( n+{{{\tilde{K}}}_{1}}\left( n \right)+{{{\tilde{K}}}_{2}}\left( n \right) \right)\left( n-1+{{{\tilde{K}}}_{1}}\left( n-1 \right)+{{{\tilde{K}}}_{2}}\left( n-1 \right) \right)}.
  \end{aligned}
$$
Therefore,
$$
\begin{aligned}
{{S}_{n-2}}=&\frac{\left( n+{{{\tilde{K}}}_{1}}\left( n \right)+{{{\tilde{K}}}_{2}}\left( n \right) \right)\left( n-1+{{{\tilde{K}}}_{1}}\left( n-1 \right)+{{{\tilde{K}}}_{2}}\left( n-1 \right) \right)}{{{{\tilde{K}}}_{3}}\left( n-1 \right){{{\tilde{K}}}_{2}}\left( n \right)}{{A}_{n}}\\
&-\frac{\left( n-1+{{{\tilde{K}}}_{1}}\left( n-1 \right)+{{{\tilde{K}}}_{2}}\left( n-1 \right) \right)}{{{{\tilde{K}}}_{2}}\left( n-1 \right)}{{A}_{n-1}}-{{\tilde{K}}_{2}}\left( 0 \right),
\end{aligned}
$$
and
$$
\begin{aligned}
{{A}_{n+1}}=&\frac{{{{\tilde{K}}}_{3}}\left( n \right){{{\tilde{K}}}_{2}}\left( n+1 \right){{{\tilde{K}}}_{2}}\left( 0 \right)}{\left( n+{{{\tilde{K}}}_{1}}\left( n \right)+{{{\tilde{K}}}_{2}}\left( n \right) \right)\left( n+1+{{{\tilde{K}}}_{1}}\left( n+1 \right)+{{{\tilde{K}}}_{2}}\left( n+1 \right) \right)}\\
&+\frac{{{{\tilde{K}}}_{3}}\left( n \right){{{\tilde{K}}}_{2}}\left( n+1 \right)}{\left( n+1+{{{\tilde{K}}}_{1}}\left( n+1 \right)+{{{\tilde{K}}}_{2}}\left( n+1 \right) \right){{{\tilde{K}}}_{2}}\left( n \right)}{{A}_{n}} \\
&+\frac{{{{\tilde{K}}}_{3}}\left( n \right){{{\tilde{K}}}_{2}}\left( n+1 \right)}{\left( n+{{{\tilde{K}}}_{1}}\left( n \right)+{{{\tilde{K}}}_{2}}\left( n \right) \right)\left( n+1+{{{\tilde{K}}}_{1}}\left( n+1 \right)+{{{\tilde{K}}}_{2}}\left( n+1 \right) \right)}\left( {{A}_{n-1}}+{{S}_{n-2}} \right).
\end{aligned}
$$
Substituting the expression of ${{S}_{n-2}}$ into that of ${{A}_{n+1}}$ yields
$$
\begin{aligned}
{{A}_{n+1}} =& \left[ \frac{{{{\tilde{K}}}_{3}}\left( n \right){{{\tilde{K}}}_{2}}\left( n+1 \right)}{\left( n+1+{{{\tilde{K}}}_{1}}\left( n+1 \right)+{{{\tilde{K}}}_{2}}\left( n+1 \right) \right){{{\tilde{K}}}_{2}}\left( n \right)} +\right.\\
&\phantom{=\;\;}\left.\frac{{{{\tilde{K}}}_{3}}\left( n \right){{{\tilde{K}}}_{2}}\left( n+1 \right)\left( n-1+{{{\tilde{K}}}_{1}}\left( n-1 \right)+{{{\tilde{K}}}_{2}}\left( n-1 \right) \right)}{{{{\tilde{K}}}_{3}}\left( n-1 \right){{{\tilde{K}}}_{2}}\left( n \right)\left( n+1+{{{\tilde{K}}}_{1}}\left( n+1 \right)+{{{\tilde{K}}}_{2}}\left( n+1 \right) \right)} \right]{{A}_{n}} \\
&+ \left\{ \frac{{{{\tilde{K}}}_{3}}\left( n \right){{{\tilde{K}}}_{2}}\left( n+1 \right)}{\left( n+{{{\tilde{K}}}_{1}}\left( n \right)+{{{\tilde{K}}}_{2}}\left( n \right) \right)\left( n+1+{{{\tilde{K}}}_{1}}\left( n+1 \right)+{{{\tilde{K}}}_{2}}\left( n+1 \right) \right)} \right.\\
&\phantom{=\;\;}\left.\times \left[ 1-\frac{\left( n-1+{{{\tilde{K}}}_{1}}\left( n-1 \right)+{{{\tilde{K}}}_{2}}\left( n-1 \right) \right)}{{{{\tilde{K}}}_{2}}\left( n-1 \right)} \right]\right\}{{A}_{n-1}}
\end{aligned}
$$
which can be rewritten as
$$
\begin{aligned}
{{A}_{n+1}}=&\frac{{{{\tilde{K}}}_{3}}\left( n \right){{{\tilde{K}}}_{2}}\left( n+1 \right)\left( n-1+{{{\tilde{K}}}_{1}}\left( n-1 \right)+{{{\tilde{K}}}_{2}}\left( n-1 \right)+{{{\tilde{K}}}_{3}}\left( n-1 \right) \right)}{{{{\tilde{K}}}_{3}}\left( n-1 \right){{{\tilde{K}}}_{2}}\left( n \right)\left( n+1+{{{\tilde{K}}}_{1}}\left( n+1 \right)+{{{\tilde{K}}}_{2}}\left( n+1 \right) \right)}{{A}_{n}} \\
&-\frac{{{{\tilde{K}}}_{3}}\left( n \right){{{\tilde{K}}}_{2}}\left( n+1 \right)\left( n-1+{{{\tilde{K}}}_{1}}\left( n-1 \right) \right)}{\left( n+{{{\tilde{K}}}_{1}}\left( n \right)+{{{\tilde{K}}}_{2}}\left( n \right) \right)\left( n+1+{{{\tilde{K}}}_{1}}\left( n+1 \right)+{{{\tilde{K}}}_{2}}\left( n+1 \right) \right){{{\tilde{K}}}_{2}}\left( n-1 \right)}{{A}_{n-1}}.
\end{aligned}
$$
Using the expression: ${{A}_{n}}=\frac{{{{\tilde{K}}}_{2}}\left( n \right){{a}_{n}}}{\sum\nolimits_{i=1}^{n}{\left( i+{{{\tilde{K}}}_{1}}\left( i \right)+{{{\tilde{K}}}_{2}}\left( i \right) \right)}}$, we further have
$$
\begin{aligned}
&\frac{{{{\tilde{K}}}_{2}}\left( n+1 \right){{a}_{n+1}}}{\prod\limits_{i=1}^{n+1}{\left( i+{{{\tilde{K}}}_{1}}\left( i \right)+{{{\tilde{K}}}_{2}}\left( i \right) \right)}}\\
&\quad=\frac{{{{\tilde{K}}}_{3}}\left( n \right){{{\tilde{K}}}_{2}}\left( n+1 \right)\left( n-1+{{{\tilde{K}}}_{1}}\left( n-1 \right)+{{{\tilde{K}}}_{2}}\left( n-1 \right)+{{{\tilde{K}}}_{3}}\left( n-1 \right) \right)}{{{{\tilde{K}}}_{3}}\left( n-1 \right){{{\tilde{K}}}_{2}}\left( n \right)\left( n+1+{{{\tilde{K}}}_{1}}\left( n+1 \right)+{{{\tilde{K}}}_{2}}\left( n+1 \right) \right)}\frac{{{{\tilde{K}}}_{2}}\left( n \right){{a}_{n}}}{\prod\limits_{i=1}^{n}{\left( i+{{{\tilde{K}}}_{1}}\left( i \right)+{{{\tilde{K}}}_{2}}\left( i \right) \right)}} \\
&\quad-\frac{{{{\tilde{K}}}_{3}}\left( n \right){{{\tilde{K}}}_{2}}\left( n+1 \right)\left( n-1+{{{\tilde{K}}}_{1}}\left( n-1 \right) \right)}{\left( n+{{{\tilde{K}}}_{1}}\left( n \right)+{{{\tilde{K}}}_{2}}\left( n \right) \right)\left( n+1+{{{\tilde{K}}}_{1}}\left( n+1 \right)+{{{\tilde{K}}}_{2}}\left( n+1 \right) \right){{{\tilde{K}}}_{2}}\left( n-1 \right)}\frac{{{{\tilde{K}}}_{2}}\left( n-1 \right){{a}_{n-1}}}{\prod\limits_{i=1}^{n-1}{\left( i+{{{\tilde{K}}}_{1}}\left( i \right)+{{{\tilde{K}}}_{2}}\left( i \right) \right)}}.
\end{aligned}
$$
Thus, we obtain
\begin{equation}\label{eq:A2a}
\begin{aligned}
{{a}_{n+1}}=&\frac{{{{\tilde{K}}}_{3}}\left( n \right)\left( n-1+{{{\tilde{K}}}_{1}}\left( n-1 \right)+{{{\tilde{K}}}_{2}}\left( n-1 \right)+{{{\tilde{K}}}_{3}}\left( n-1 \right) \right)}{{{{\tilde{K}}}_{3}}\left( n-1 \right)}{{a}_{n}} \\
&-{{{\tilde{K}}}_{3}}\left( n \right)\left( n-1+{{{\tilde{K}}}_{1}}\left( n-1 \right) \right){{a}_{n-1}},
\end{aligned}
\end{equation}
where $n\ge 2$, ${{a}_{1}}={{\tilde{K}}_{3}}\left( 0 \right)$, and ${{a}_{2}}={{\tilde{K}}_{3}}\left( 1 \right)\left( {{{\tilde{K}}}_{3}}\left( 0 \right)+{{{\tilde{K}}}_{2}}\left( 0 \right) \right)$. In a similar way, we can prove
\begin{equation}\label{eq:A2b}
\begin{aligned}
{{b}_{n+1}}=&\frac{{{{\tilde{K}}}_{3}}\left( n \right)\left( n-1+{{{\tilde{K}}}_{1}}\left( n-1 \right)+{{{\tilde{K}}}_{2}}\left( n-1 \right)+{{{\tilde{K}}}_{3}}\left( n-1 \right) \right)}{{{{\tilde{K}}}_{3}}\left( n-1 \right)}{{b}_{n}} \\
&-{{{\tilde{K}}}_{3}}\left( n \right)\left( n-1+{{{\tilde{K}}}_{1}}\left( n-1 \right) \right){{b}_{n-1}},
\end{aligned}
\end{equation}
where $n\ge 2$, ${{b}_{1}}={{\tilde{K}}_{3}}\left( 0 \right)$, and ${{b}_{2}}={{\tilde{K}}_{3}}\left( 1 \right)\left( {{{\tilde{K}}}_{3}}\left( 0 \right)+{{{\tilde{K}}}_{2}}\left( 0 \right)+{{{\tilde{K}}}_{1}}\left( 0 \right) \right)$.

Second, we rewrite \eqref{eq:A2a} as
\begin{equation}\label{eq:A3}
\begin{aligned}
&\frac{{{{\tilde{K}}}_{3}}\left( n-1 \right)}{{{{\tilde{K}}}_{3}}\left( n \right)}\left[ {{a}_{n+1}}-{{{\tilde{K}}}_{3}}\left( n \right){{a}_{n}} \right]=\left[ n-1+{{{\tilde{K}}}_{1}}\left( n-1 \right)+{{{\tilde{K}}}_{2}}\left( n-1 \right) \right] \\
&\qquad \quad\times \left[ {{a}_{n}}-\frac{n-1+{{{\tilde{K}}}_{1}}\left( n-1 \right)}{n-1+{{{\tilde{K}}}_{1}}\left( n-1 \right)+{{{\tilde{K}}}_{2}}\left( n-1 \right)}{{{\tilde{K}}}_{3}}\left( n-1 \right){{a}_{n-1}} \right],
\end{aligned}
\end{equation}
where $n=2,3,\cdots $. Note that
$$
\begin{aligned}
{{\left[ {{a}_{n}}-\frac{n-1+{{{\tilde{K}}}_{1}}\left( n-1 \right)}{n-1+{{{\tilde{K}}}_{1}}\left( n-1 \right)+{{{\tilde{K}}}_{2}}\left( n-1 \right)}{{{\tilde{K}}}_{3}}\left( n-1 \right){{a}_{n-1}} \right]}_{n=2}}\\
\qquad \quad ={{\tilde{K}}_{3}}\left( 1 \right)\left[ \frac{{{{\tilde{K}}}_{2}}\left( 1 \right)}{1+{{{\tilde{K}}}_{1}}\left( 1 \right)+{{{\tilde{K}}}_{2}}\left( 1 \right)}{{{\tilde{K}}}_{3}}\left( 0 \right)+{{{\tilde{K}}}_{2}}\left( 0 \right) \right]>0.
\end{aligned}
$$
By the mathematical induction, we can prove
\begin{equation}\label{eq:A4}
{{a}_{n}}-\frac{n-1+{{{\tilde{K}}}_{1}}\left( n-1 \right)}{n-1+{{{\tilde{K}}}_{1}}\left( n-1 \right)+{{{\tilde{K}}}_{2}}\left( n-1 \right)}{{\tilde{K}}_{3}}\left( n-1 \right){{a}_{n-1}}>0,
\end{equation}
for $n=2,3,\cdots $. Thus, it follows from \eqref{eq:A3} that
\begin{equation}\label{eq:A5}
{{a}_{n+1}}>{{\tilde{K}}_{3}}\left( n \right){{a}_{n}},
\end{equation}
where $n=2,3,\cdots $. It also follows from \eqref{eq:A3} that
$$
\frac{{{{\tilde{K}}}_{3}}\left( n-1 \right)}{{{{\tilde{K}}}_{3}}\left( n \right)}\left[ {{a}_{n+1}}-{{{\tilde{K}}}_{3}}\left( n \right){{a}_{n}} \right]>\left[ n-1+{{{\tilde{K}}}_{1}}\left( n-1 \right)+{{{\tilde{K}}}_{2}}\left( n-1 \right) \right]\left[ {{a}_{n}}-{{{\tilde{K}}}_{3}}\left( n-1 \right){{a}_{n-1}} \right],
$$
from which we can have
\begin{equation}\label{eq:A6}
\left[ {{a}_{n+1}}-{{{\tilde{K}}}_{3}}\left( n \right){{a}_{n}} \right]>\left[ {{a}_{2}}-{{{\tilde{K}}}_{3}}\left( 1 \right){{a}_{1}} \right]\prod\limits_{k=2}^{n}{{{A}_{k}}},
\end{equation}
where ${{A}_{n}}=\frac{{{{\tilde{K}}}_{3}}\left( n \right)}{{{{\tilde{K}}}_{3}}\left( n-1 \right)}\left[ n-1+{{{\tilde{K}}}_{1}}\left( n-1 \right)+{{{\tilde{K}}}_{2}}\left( n-1 \right) \right]$, and $n=2,3,\cdots $. According to Lemma \ref{lm:1}, we can show
\begin{equation}\label{eq:A7}
{{a}_{n+1}}>\left[ \prod\limits_{i=2}^{n}{{{{\tilde{K}}}_{3}}\left( i \right)} \right]\left\{ {{a}_{2}}+\left[ {{a}_{2}}-{{{\tilde{K}}}_{3}}\left( 1 \right){{a}_{1}} \right]\sum\limits_{i=2}^{n}{\prod\limits_{k=2}^{i}{\frac{k-1+{{{\tilde{K}}}_{1}}\left( k-1 \right)+{{{\tilde{K}}}_{2}}\left( k-1 \right)}{{{{\tilde{K}}}_{3}}\left( k-1 \right)}}} \right\},
\end{equation}
where $n=2,3,\cdots $. Appareantly, ${{{a}_{n}}}/{n!}\;$ tends to infinity if ${{\tilde{K}}_{1}}\left( n \right)+{{\tilde{K}}_{2}}\left( n \right)>2$ for $n=2,3,\cdots $.

Third, \eqref{eq:A3} can be rewritten as
$$
{{x}_{n+1}}=\frac{{{{\tilde{K}}}_{3}}\left( n \right)\left( n-1+{{{\tilde{K}}}_{1}}\left( n-1 \right)+{{{\tilde{K}}}_{2}}\left( n-1 \right) \right)}{{{{\tilde{K}}}_{3}}\left( n-1 \right)\left( n+1 \right)}{{x}_{n}}+\frac{{{{\tilde{K}}}_{3}}\left( n \right)}{n+1}\left[ {{x}_{n}}-\frac{n-1+{{{\tilde{K}}}_{1}}\left( n-1 \right)}{n}{{x}_{n=1}} \right],
$$
where ${{x}_{n}}={{a}_{n}}/{n!}$ and $n=2,3,\cdots $. If ${{\tilde{K}}_{1}}\left( n-1 \right)+{{\tilde{K}}_{2}}\left( n-1 \right)\le 2$ and ${{\tilde{K}}_{3}}\left( n \right)\le M$, then when $n$ is sufficiently large, we have
$$
{{x}_{n+1}}-{{x}_{n}}\approx \frac{{{{\tilde{K}}}_{3}}\left( n \right)}{n}\left( {{x}_{n}}-{{x}_{n=1}} \right)\Rightarrow {{x}_{n+1}}-{{x}_{n}}\approx \left( {{x}_{2}}-{{x}_{1}} \right)\prod\limits_{k=2}^{n}{\frac{{{{\tilde{K}}}_{3}}\left( k \right)}{k}}.
$$
Thus, we can obtain ${{x}_{n+1}}\approx {{x}_{2}}+\left( {{x}_{2}}-{{x}_{1}} \right)\sum\limits_{m=2}^{n}{\frac{1}{m!}\prod\limits_{k=2}^{m}{{{{\tilde{K}}}_{3}}\left( k \right)}}$, implying that ${{x}_{n}}$ is convergent as $n\to \infty $. Numerical results are demonstrated in Figure \ref{f1}.

\bibliography{library}
\bibliographystyle{plain}

\end{document}